%% file: best-match-main.tex
\documentclass[11pt, twoside]{article}
\pdfoutput=1
\usepackage[margin=1in]{geometry}

\usepackage{endnotes}
\let\footnote=\endnote

\usepackage{natbib}
 \bibpunct[, ]{(}{)}{,}{a}{}{,}%
 \def\bibfont{\small}%

\usepackage{amsmath,amssymb,amsthm}
\usepackage{bm,dsfont}
\usepackage[inline]{enumitem}
\usepackage{algorithm,algpseudocode}
\usepackage{wrapfig,subfig,makecell}
\usepackage{mathtools}
\usepackage{color}
\usepackage[colorlinks=true,breaklinks=true,bookmarks=true,urlcolor=blue,
     citecolor=blue,linkcolor=blue,bookmarksopen=false,draft=false]{hyperref}
\usepackage[capitalise]{cleveref}
\usepackage[normalem]{ulem}

\def\EE{{\mathbb{E}}}\def\PP{{\mathbb{P}}}

\def\setN{\mathbb{N}}

\newcounter{const-no}
\newcommand{\const}[1]{\refstepcounter{const-no}\label{#1}}
\newcounter{event-no}
\newcommand{\event}[1]{\refstepcounter{event-no}\label{#1}}
\newcommand*\diff{\mathop{}\!\mathrm{d}}
\newcommand\numberthis{\addtocounter{equation}{1}\tag{\theequation}}

\newcommand{\ind}[1]{\mathds{1}\left\lbrace #1 \right\rbrace}
\newcommand\givenbase[1][]{\:#1\lvert\:}
\let\given\givenbase

\DeclarePairedDelimiterX\Basics[1](){\let\given\sgiven #1}

\newtheorem{theorem}{Theorem}
\newtheorem{corollary}[theorem]{Corollary}
\newtheorem{proposition}[theorem]{Proposition}
\newtheorem{lemma}[theorem]{Lemma}

\newtheorem{remark}{Remark}
\newtheorem{definition}{Definition} 
\newtheorem{assumption}{Assumption} 

\newtheorem{fact}{Fact}

\usepackage{authblk}

\date{}

\title{Learning Unknown Service Rates in Queues: A Multi-Armed Bandit Approach}

\author[1]{Subhashini Krishnasamy}
\author[1]{Rajat Sen}
\author[2]{Ramesh Johari}
\author[1]{Sanjay Shakkottai}
\affil[1]{The University of Texas at Austin}
\affil[2]{Stanford University}

\begin{document}

\maketitle

\begin{abstract}
Consider a queueing system consisting of multiple servers.  Jobs arrive over time and enter a queue for service; the goal is to minimize the size of this queue.  At each opportunity for service, at most one server can be chosen, and at most one job can be served.  Service is successful with a probability (the {\em service probability}) that is {\em a priori unknown} for each server.  An algorithm that knows the service probabilities (the ``genie'') can always choose the server of highest service probability.  We study algorithms that learn the unknown service probabilities.  Our goal is to minimize {\em queue-regret}:  the (expected) difference between the queue-lengths obtained by the  algorithm, and those obtained by the ``genie.''

Since queue-regret cannot be larger than classical regret, results for
the standard multi-armed bandit problem give algorithms for which queue-regret
increases no more than logarithmically in time.  Our paper shows
surprisingly more complex behavior.  In particular, as long as the
bandit algorithm's queues have relatively long regenerative cycles,
queue-regret is similar to cumulative regret, and scales (essentially)
logarithmically.  However, we show that this ``early stage'' of the
queueing bandit eventually gives way to a ``late stage'', where the
optimal queue-regret scaling is $O(1/t)$.  We demonstrate an algorithm that
(order-wise) achieves this asymptotic queue-regret in the late stage. 
Our results are developed in a more general model that allows for multiple job classes as well.
\end{abstract}


\input{intro2}
%
\input{related}

\input{sys-model}

\input{late-stage}
\input{early-stage}
\input{simulationsv2}

\input{conclusion}

\input{proof-ub}
\input{proof-lb}

\theendnotes

\section*{Acknowledgments}
This work is partially supported by NSF Grants CNS-1161868, CNS-1343383, CNS-1320175, ARO grants W911NF-15-1-0227, W911NF-14-1-0387 and the US DoT supported D-STOP Tier 1 University
Transportation Center.


\bibliographystyle{informs2014} 
\bibliography{queue-bandits} 

%

\end{document}

%% file: intro2.tex
{\section{Introduction}

Stochastic multi-armed bandits (MAB) have a rich history in sequential decision making (\citet{gittins1979bandit,mahajan2008multi,bubeck-bianchi12survey-ucb-algo}). In its simplest form, at each discrete time step the decision maker must choose a single arm from a collection of $K$ arms. A random binary reward (i.e., a Bernoulli random variable taking value 0 or 1) is accrued each time an arm is pulled; if a reward is received, we refer to the outcome as a ``success.''  The probabilities of success vary across arms, and are unknown a priori. The MAB problem is to determine which arm to choose at each time in order to minimize the {\em cumulative expected regret}: the cumulative loss of expected reward when compared to a genie that has knowledge of the arm success probabilities.  

In this paper\footnote{An earlier version of this work appeared in the Proceedings of the Thirtieth Annual Conference on Neural Information Processing Systems (NIPS), 2016 (\citet{krishnasamy-etal16regret-queue}).}, we study a variant of this problem motivated by {\em queueing} applications, where success probabilities may be unknown.
\label{edit:prob-descr}  

To fix ideas, consider the problem of scheduling in a discrete-time queueing system with a single queue and multiple servers. Jobs arrive to the system and are stored in the queue until they are successfully served. At any time, only one server can be active and can serve at most one job in the queue. The active server is either successful (in which case the job departs the system) or unsuccessful (in which  case the job remains in the queue); success probabilities vary across servers.  We consider this problem in the case where the success probabilities (a.k.a. service probabilies) are not known {\em a priori}.

We view this problem as a variant of the basic MAB model, where each arm is now a {\em server} that can serve a waiting job.  From this perspective, the stochastic reward described above is equivalent to {\em service}, which takes binary values (1 or 0) depending on whether the job was successfully served or not. At any time the queue is empty, the only possible reward is 0. We observe that the basic MAB model is not suitable for queueing applications as it fails to capture an essential feature of service in many settings: in a queueing system, {\em jobs wait until they complete service}.

Such systems are {\em stateful}: when the chosen arm results in zero reward, the job being served remains in the queue. On the other hand, when there are no jobs waiting in the queue, there can be no accrual of reward. This makes it essential in this model to track the number of jobs waiting in the queue to be served.  The {\em queue length}, which is the difference between the cumulative number of arrivals and departures, is the most common measure of the quality of the service strategy being employed.

Our paper is motivated by combining two well-established lines of research: first, the ubiquity of queueing as a means to model service systems; and second, the prevalence of MAB models as a way to study the balance between learning and performance optimization in a wide range of settings.  Our paper  develops an MAB model to study the role of unknown service rates on performance in scheduling systems.
This problem clearly has the explore-exploit tradeoff inherent in the basic MAB problem: since the service probabilities across different servers are unknown, there is a tradeoff between learning (\textit{exploring}) the different servers and scheduling (\textit{exploiting}) the most promising server from past observations.
 For brevity, we refer to this problem as the {\em queueing bandit}.
Since the queue length is simply the difference between the cumulative number of arrivals and departures, the natural notion of regret here is to compare the expected queue length under a bandit algorithm with the corresponding one under a genie policy that always chooses the arm with the highest expected reward.

\label{edit:full-problem}
Our results are actually developed for a switch network, i.e., a multi-class parallel-server system; this is a more general model than the single queue system described above. Specifically, we allow the possibility that there are $U$ distinct classes of jobs with a queue for each class.  Each of the job queues may be served by any of the $K$ ($K \geq U$) servers, but the service probabilities are heterogeneous across queue-server pairs (also referred to as {\em links}).  At each time, at most one queue can be assigned to each server, and at most one server can be assigned to each queue; thus in each time slot, scheduling amounts to choosing a matching in the complete bipartite graph between queues and servers. Arrivals to the queue and possible service offered by the links follow a product Bernoulli distribution, and are i.i.d.\ across time slots.  A job remains in the queue if not successfully served; 
service on the same job can be attempted by different servers at different times. Statistical parameters (service probabilities) corresponding to the service distributions are considered unknown. 

In the single queue setting, i.e., when $U = 1$, we assume that there is at least one server that has a service probability higher than the arrival probability.  This ensures that the optimal ``genie'' policy---i.e., the policy with full {\em a priori} knowledge of all the service probabilities---is stable. When $U > 1$, we assume that for the true service probabilities, there exists a single unique matching (the {\em optimal matching}) which is strictly better than all other matchings for all queues.
For $U > 1$, this assumption ensures that the optimal ``genie'' policy is one that selects the  optimal matching in every time slot.

Let $\bm{Q}(t)$ be the queue length vector at time $t$ under a given bandit algorithm, and let $\bm{Q}^*(t)$ be the corresponding queue length vector under the ``genie'' policy that always chooses the optimal matching.  We define the {\em queue-regret} vector $\bm{\Psi}(t)$ as the difference in expected queue lengths for the two policies:  
\begin{align}
\label{eqn:q-regret}
\bm{\Psi}(t) := \EE\left[ \bm{Q}(t) - \bm{Q}^*(t) \right].
 \end{align}
Note that the difference in queue-lengths is equal to the difference in the total number of departures. Interpreting a departure as reward $1$, the queue-regret $\bm{\Psi}(t)$ is the expected difference of the accumulated rewards. This is similar to the traditional MAB regret but with a slight difference: here, it is possible to accrue reward only if the queue is non-empty. 
As an example, in a single queue setting, suppose the number of waiting jobs is infinite; in that case, the queue-regret is identical to the regret in the basic MAB problem.
 Our goal is to develop bandit algorithms that lead to small queue-regret for each $u \in [U]$ at a finite time $t$. To develop some intuition, we compare this with the standard stochastic MAB problem; we focus on the setting of a single queue ($U = 1$) for ease of exposition, but the same insights hold 
in the general case of a switch network with a unique optimal matching when we consider each component of the queue-regret individually.  
For the standard MAB problem, well-known algorithms such as UCB, KL-UCB, and Thompson sampling achieve a cumulative regret of $O((K-1)\log t)$ at time $t$ (\citet{auer2002finite, garivier2011kl, agrawal-goyal12thompson}), and this result is essentially tight (\citet{lai1985asymptotically}).
In the queueing bandit, it can be shown that the queue-regret cannot be any higher than the traditional regret (where a reward is accrued at each time whether a job is present or not).  This leads to an upper bound of $O((K-1)\log t)$ for the queue-regret.

However, this upper bound does not tell the whole story for the queueing bandit: we show that there are two ``stages'' to the queueing bandit.  If the arrival probability of a queue is higher than the service probabilities of all but the corresponding best server, the bandit algorithm is unable to even stabilize the queue in the {\em early} stage -- 
i.e., on average, the queue length increases over time and is continuously backlogged; therefore the queue-regret grows  similarly to the traditional regret.  Once the algorithm is able to stabilize the queue --- i.e., in the {\em late} stage --- a dramatic shift occurs in the behavior of the queue-regret.  A stochastically stable queue goes through {\bf regenerative cycles} -- a random cyclical behavior where queues build-up over time, then empty, and the cycle repeats. The associated recurring``zero-queue-length'' epochs means that sample-path queue-regret essentially ``resets'' at (stochastically) regular intervals; i.e., the sample-path queue-regret becomes non-positive at these time instants.  Thus the queue-regret should fall over time, as the algorithm learns.

\subsection{Contributions}
Our main results provide lower bounds on  each component of the queue-regret vector in both the early and late stages, as well as algorithms that essentially match these lower bounds.  Here, we discuss these results in the context of a single queue, but prove them in the more general switch network context described above.

Below we describe our main contributions.


\begin{description}
\item [{\bf The late stage:}]  We first consider what happens to the queue-regret as $t \to \infty$.  As noted above, a reasonable intuition for this regime comes from considering a standard bandit algorithm, but where the sample-path queue-regret ``resets'' at time points of regeneration. This is inexact since the optimal queueing system and bandit queueing system may not regenerate at the same time point; nevertheless, the intuition proves correct.  In this case, the queue-regret is approximately a (discrete) {\em derivative} of the cumulative regret.  Since the optimal cumulative regret scales like $\log t$, asymptotically the optimal queue-regret should scale like $1/t$.
Indeed, we show that the queue-regret for $\alpha$-consistent policies (please see \cref{defn:consistent,defn:consistent-multiq} in \cref{sec:late-stg} for the definition of an $\alpha$-consistent policy) is at least $C/t$ infinitely often, where $C$ is a constant independent of $t$.  Further, we introduce scheduling algorithms called Q-UCB and Q-ThS, which are variants of the UCB1 and Thompson sampling algorithms tailored to the queueing bandit. We show an asymptotic regret upper bound of $O\left(\mathrm{poly}(\log t)/t \right)$ for both these algorithms, thus matching the lower bound up to poly-logarithmic factors in $t$.  The key feature of Q-UCB and Q-ThS is that they use {\em forced exploration}: they exploit the fact that the queue regenerates regularly to explore more systematically and aggressively.

\item [{\bf The early stage:}]  The preceding discussion might suggest that an algorithm that explores aggressively would dominate any algorithm that balances exploration and exploitation.  However, an overly aggressive exploration policy will preclude the queueing system from ever stabilizing, which is {\em necessary} to induce the regenerative cycles that lead the system to the late stage. 
To even enter the late stage, therefore, we need an algorithm that exploits enough to actually stabilize the queue (i.e., we must choose good arms sufficiently often so that the service probability exceeds the arrival probability). 

We refer to the early stage of the system, as noted above, as the period before the algorithm has learned to stabilize the queues.  For a {\em heavily loaded system, where the arrival probability approaches the service probability of the optimal server,} we show a lower bound of $\Omega(\log t/\log \log t)$ on the queue-regret in the early stage.  Thus up to a $1/\log \log t$ factor, the early stage regret behaves similarly to the cumulative regret (which scales like $\log t$).  The heavily loaded regime is a natural asymptotic regime in which to study queueing systems, and has been extensively employed in the literature; see, e.g., \citet{whitt1974heavy, kushner2013heavy} for surveys.
\end{description}

Our results constitute the first insight into the behavior of regret in this queueing setting; as emphasized, it is quite different from minimization of cumulative regret in the standard MAB problem.  The preceding discussion highlights why minimization of queue-regret presents a subtle learning problem.  On one hand, if the queue has been stabilized, the presence of regenerative cycles allows us to establish that queue-regret must eventually decay to zero at rate $1/t$ under an optimal algorithm (the late stage).  On the other hand, to actually have regenerative cycles in the first place, a learning algorithm needs to exploit enough to actually stabilize the queue (the early stage).  Our analysis not only characterizes regret in both regimes, but also essentially exactly characterizes the transition point between the two regimes.  In this way the queueing bandit is a remarkable new example of the tradeoff between exploration and exploitation.


%


%% file: related.tex
\section{Related Work}
\label{sec:related}

\begin{description}

\item[{\bf MAB algorithms}:] MAB models have been widely used in the past as a paradigm for various sequential decision making problems in industrial manufacturing, communication networks, clinical trials, online advertising and webpage optimization, and other domains requiring resource allocation and scheduling; see, e.g., \cite{gittins1979bandit,mahajan2008multi,bubeck-bianchi12survey-ucb-algo}.

The classical MAB problem is based on the stochastic MAB model where the rewards of the arms are i.i.d.\ across time. The goal is to minimize the expected cumulative loss of reward relative to a ``genie'' policy that always chooses best arm; this objective is called {\em regret}.  There is a vast body of literature which aims at getting the best  possible \emph{finite time} lower bounds for regret and  designing algorithms that match the lower bound (see \cite{bubeck-bianchi12survey-ucb-algo} for a survey). \cite{lai1985asymptotically} prove a lower bound of $O(\log t)$ with the scaling constant depending on the mean rewards of the best and the second best arm.

A variant of the MAB problem that is related to our problem is the Markovian bandit in which the state of the each arm evolves in a Markovian fashion, and the reward drawn at each time is a function of the state of the selected arm. Two kinds of model for the underlying Markov decision process (MDP) have been studied.  The \emph{rested} bandit model, introduced by \cite{gittins1979bandit}, assumes that the state of an arm is frozen at a time step unless it is pulled. \cite{whittle1988restless} introduced the \emph{restless} bandit model, where the state of every arm can change at each step, regardless of which arm is pulled.

 There are two lines of research which explore the Markovian MAB problem. The traditional approach has been to assume that the statistical parameters associated with the Markov chain of the arms are perfectly known to the decision maker. The aim is to optimize the \emph{infinite horizon} discounted or average reward of the corresponding MDP. Papers that study these problems typically propose \emph{index policies} fashioned after Gittin's index (\cite{gittins1979bandit}) and Whittle's index (\cite{whittle1988restless}) that are computationally more efficient than solving Bellman equations; see \cite{mahajan2008multi,gittins11index-policies} for a broad survey. These achieve approximately optimal infinite horizon cost.

A more recent line of research investigates the \emph{reinforcement learning} problem, where the statistical parameters of the MDP are not known a priori. This assumption adds to the optimization problem the challenge of learning the transition structure, and thus presents an explore-exploit dilemma as in the stochastic MAB problem. For the restless bandit model with finite state space and action space, \cite{jaksch2010near,ortner2014regret} prove finite time regret bounds of $O(\sqrt{t})$ where the scaling constant depends on the size of the state space and action space and the structure of the MDP. They also show that the regret bound scales as $O(\log t)$ with a scaling constant that depends on the gap between the average costs of the best and second best policy.

Although the queueing bandit problem studied in this paper has an underlying Markovian structure, the challenge is essentially different from either rested or restless bandits. Unlike the rested or restless bandits where the state of each arm evolves according to an independent Markov chain, the Markovian structure in this problem is captured by the state of the system (here queue-lengths) which is a complex function of the external arrival process, the decision rule and the rewards of the selected arm in the previous time slots. Nevertheless, it presents the same kind of exploration-exploitation dilemma as the above MAB problems when the statistical parameters are not known a priori.

%
%
%

\item[{\bf Bandits for queues}:] There is a growing body of literature on the application of bandit models to queueing and scheduling problems -- see \cite{nino2007dynamic,mahajan2008multi,gittins11index-policies,LarranagaAV16restless-bandits}. Specifically, \cite{cox1961queues,buyukkoc1985cmu,van1995dynamic,nino2006marginal,jacko2010restless,lott2000optimality,AyestaJN17} use bandit models to develop algorithms for various types of scheduling problems in queueing networks. The primary difference between these models and ours is that they assume a priori knowledge of the statistical parameters, while our focus is on learning these parameters.  In these papers, the goal is to solve a stochastic scheduling MDP, and the problem is embedded in the bandit framework as a solution approach.  Their aim is to optimize infinite-horizon costs (i.e., statistically steady-state behavior, where the focus typically is on conditions for optimality of index policies), whereas we focus on the analysis of finite time regret. 
Further, the models do not typically consider user-dependent server statistics. 

Several other problems in queueing systems such as routing and admission control have been studied in the bandit framework (\cite{nino2012admission,AvrachenkovADJ13index-policy-TCP}), again assuming known statistics. Here too, the focus is on designing index policies for these models and showing approximate optimality in steady state.

\item[{\bf Switch scheduling}:] \label{item:switch-nw} Scheduling in switch networks has received a great deal of attention in the last two decades; see, e.g., \cite{sriyin14} for a survey.  Notably, many of the scheduling algorithms (e.g., queue-length-based backpressure scheduling algorithms) can yield near-optimal performance in the absence of information about arrival rates; but these algorithms still require information about server availability and capacity.  By contrast, our work focuses on learning about all unknown aspects of the environment, as needed to deliver small overall queue lengths. 
Finally, the problem of identifying the right matchings in a bipartite graph has been formulated as a special case of the combinatorial/linear bandit problem (\cite{gai2012combinatorial,cesa2012combinatorial}, {\color{red}\cite{combes2015combinatorial,degenne2016combinatorial}}), but with a generic reward structure and not in the context of queue scheduling as in our case.

\end{description}

%% file: sys-model.tex
\section{Problem Setting}
\label{sec:sys-model}

In this section, we first introduce a baseline model of a discrete-time network with a single queue and $K$ servers.  We discuss our main results in the context of this model, but ultimately formally state and prove our main theorems in the context of a somewhat more general setting with multiple queues and multiple servers.  By focusing on the single queue model in our presentation, we are able to better elucidate the main characterizations of regret, including the distinctions between the early and late stages.

\subsection{A Single Queue Model}

Formally, suppose that arrivals to the queue and service offered by the servers are according to a product Bernoulli distribution, i.i.d.\ across time slots, with arrival probability $\lambda$  and service probabilities given by the vector $\bm{\mu} = [\mu_{k}]_{k \in [K]}$. Let the highest service probability among all the servers be denoted by $\mu^*$ and the lowest by $\mu_{min}$. We assume that the system can be stabilized if the best server is known a priori, i.e.,  $\lambda < \mu^*$.

The scheduling decision at any time $t$ is based on past observations corresponding to the services obtained for the scheduled servers until time $t-1$. Statistical parameters corresponding to the service distributions are considered unknown. 

	The queue evolution for the single queue system can be described as follows. Let $\kappa(t)$ be the server scheduled at time $t$. Also, let $R_{k}(t)$ be the service offered to the queue by server $k$ and $S(t)$ denote the service offered by server $\kappa(t)$ at time $t$. 
If $A(t)$ denotes the (binary) arrival at time $t$, then the queue-length at time $t$ is given by:
\begin{align*}
Q(t) = \left( Q(t-1) + A(t) - S(t) \right)^+.
\end{align*}

We analyze the performance of a scheduling algorithm with respect to queue-regret as a function of time and key system parameters, particularly:
\begin{enumerate}[label=(\alph*)]
\item the load on the system $\epsilon := \mu^* - \lambda$, and
\item the difference between the service probabilities of the best and the next best servers $\Delta := \mu^* - \max_{k \neq k^*} \mu_k$. 
\end{enumerate}

\subsection{The General Model: A Multi-Queue Switch Network}
\label{subsec:sys-model-multiq}
As mentioned, we prove the results in Sections~\ref{sec:late-stg} and \ref{sec:early-stg} for a more general problem that deals with a discrete-time stochastic switch network described as follows. The multi-queue switch network consists of $U$ queues and $K$ servers, where $U \leq K$. The queues and servers are indexed by $u = 1, \dots, U$ and $k = 1, \dots, K$ respectively.  Arrivals to queues and service offered by the links are according to product Bernoulli distribution and i.i.d.\ across time slots. The arrival probabilities are given by the vector $\bm{\lambda} = \left( \lambda_u \right)_{u \in [U]}$ and the service probabilities by the matrix $\bm{\mu} = [\mu_{uk}]_{u \in [U], k \in [K]}$.

We require the following notational definitions in our technical development:
\begin{align*}
\mu^*_u &:= \max_{k \in [K]} \mu_{uk},\ \ u \in [U];\\
k^*_u &:= \arg\max_{k \in [K]} \mu_{uk},\ \ u \in [U];\\
\epsilon_u &:= \mu^*_u - \lambda_u,\ \ u \in [U];\\
\Delta_{uk} &:= \mu^*_u - \mu_{uk},\ \ u \in [U], k \in [K];\\
\Delta &:= \min_{u \in [U], k \not\in k^*_u}\Delta_{uk};\\
\mu_{min} &:= \min_{u \in [U], k \in [K]} \mu_{uk};\\
\mu^* &:= \max_{u \in [U], k \in [K]} \mu_{uk};\\
\lambda_{min} &:= \min_{u \in [U]} \lambda_u.
\end{align*}
In any time slot, each server can serve at most one queue and each queue can be served by at most one server. The task is to schedule, in every time slot, a matching in the complete bipartite graph between queues and servers.  Let $\kappa_u(t)$ denote the server that is assigned to queue $u$ at time $t$. Therefore, the vector $\pmb{\kappa}(t) = (\kappa_u(t)_{u \in [U]})$ gives the matching scheduled at time $t.$ Other important notation for the multi-queue setting can be found in Table~\ref{tab:notation}. 
\paragraph*{Notation:} Boldface letters are used to denote vectors or matrices and the corresponding non-bold letters to denote their individual components. Also, the notation 
\begin{enumerate}[label=(\roman*)]
\item $\mathds{1}\{\cdot\}$ is used to denote the indicator function, and
\item for any $k \in \setN$, $\log^k(\cdot)$ is used to denote $(\log(\cdot))^k$. 
\end{enumerate}


\subsubsection*{Unique Optimal Matching}
We focus on a simple special case of the above switch scheduling system.  In particular, we assume for every queue, there is a unique optimal server with the maximum service probability for that queue.  Further, we assume that the optimal queue-server pairs form a matching in the complete bipartite graph between queues and servers, that we call the {\em optimal matching}; and that this optimal matching stabilizes every queue. 

Formally, we make the following assumption on the switch scheduling system.
\begin{assumption}[Unique Optimal Matching]
\label{ass:best-match}
There is a {\em unique optimal matching}, i.e.:
\begin{enumerate}
\item There is a unique optimal server for each queue: $k^*_u$ is a singleton, i.e., $\Delta_{uk} > 0$ for $k \neq k^*_u$, for all $u$, \label{item:strict-best}
\item The optimal queue-server pairs for a matching: For any $u' \neq u$, $k^*_u \neq k^*_{u'}$. \label{item:best-match}
\end{enumerate}
\end{assumption}
The assumption of a unique optimal matching essentially means that the queues and servers are solving a pure coordination problem.  It is most applicable in settings where there is strong horizontal differentiation across jobs and servers.  For example, in a crowdsourcing system where the ``servers'' are worker types completing different types of jobs, the unique optimal matching assumption implies there is a unique worker type best suited to each type of job.

We evaluate the performance of scheduling policies against the policy that schedules the optimal matching in every time slot. Let $\mathbf{Q}(t)$ be the queue-length vector at time $t$ under our specified algorithm, and let $\mathbf{Q}^*(t)$ be the corresponding vector under the optimal policy.  We define \textit{regret} as the difference in mean queue-lengths for the two policies. That is, the regret (vector) is given by: $\bm{\Psi}(t) := \EE\left[ \mathbf{Q}(t) - \mathbf{Q}^*(t) \right]$.

 Throughout, when we evaluate queue-regret, we do so under the assumption that the system is ``stable'' under the optimal policy, i.e., the policy that chooses the optimal matching at each time step. With Bernoulli distributions for the arrival and service processes and under \cref{ass:best-match}, most commonly used notions of stability \cite{neely2010stability}, including the existence of a steady-state distribution, are equivalent to the following condition:
\begin{assumption}[Stability]
$\epsilon_u  > 0$ for all $u \in [U]$.
\end{assumption}
In fact, the above condition guarantees more than just the existence of a steady-state distribution for the queue-length process  $\{\mathbf{Q}^*(t)\}$; in fact, the condition ensures that the process is geometrically ergodic. 

We also assume that the queueing system starts in the steady state distribution of the system induced by the optimal policy.
\begin{assumption}[Initial State]
\label{ass:initial-state}
The queueing system starts with an initial state $\mathbf{Q}(0)$ distributed according to the stationary distribution of $\mathbf{Q}^*(t)$, which we denote $\pi_{(\bm{\lambda}, \bm{\mu}^*)}$.
\end{assumption}
 This assumption is made largely to ease the technical exposition.  Throughout the paper, we discuss how the assumption may be weakened for each of our results.

\begin{table}
\renewcommand{\arraystretch}{}
\caption{General Notation.\label{tab:notation}}

\centering
\begin{tabular}{||c | c||}
\hline
\bfseries Symbol & \bfseries Description	\\
\hline\hline

$\lambda_u$ & Probability of arrival to queue $u$	\\ \hline
$\lambda_{min}$ & Minimum arrival probability across all queues	\\ \hline
$A_u(t)$ & Arrival at time $t$ to queue $u$	\\ \hline
$\mu_{uk}$ & Service probability of server $k$ for queue $u$	\\ \hline
$R_{uk}(t)$ & Service offered by server $k$ to queue $u$ at time $t$	\\ \hline
$k^*_u$ & Best server for queue $u$	\\ \hline
$\mu^*_u$ & Best service probability for queue $u$	\\ \hline
$\mu^*$ & Maximum service probability across all links	\\ \hline
$\mu_{min}$ & Minimum service probability across all links	\\ \hline
$\Delta$ & \makecell{Minimum (among all queues) difference \\ between	the best and second best servers}	\\ \hline
$\kappa_u(t)$ &	server assigned to queue $u$ at time $t$	\\ \hline
$S_u(t)$ & \makecell{Potential service provided by server \\ assigned to queue $u$ at time $t$}	\\ \hline
$Q_u(t)$ & queue-length of queue $u$ at time $t$ \\ \hline
$Q^*_u(t)$ & \makecell{queue-length of queue $u$ at time $t$ \\ for the optimal strategy}	\\ \hline
$\Psi_u(t)$ & Regret for queue $u$ at time $t$ \\ \hline
\hline
\end{tabular}
\end{table}

%% file: late-stage.tex
\section{The Late Stage}
\label{sec:late-stg}

\begin{figure}
\begin{center}
\includegraphics[width=0.7\textwidth]{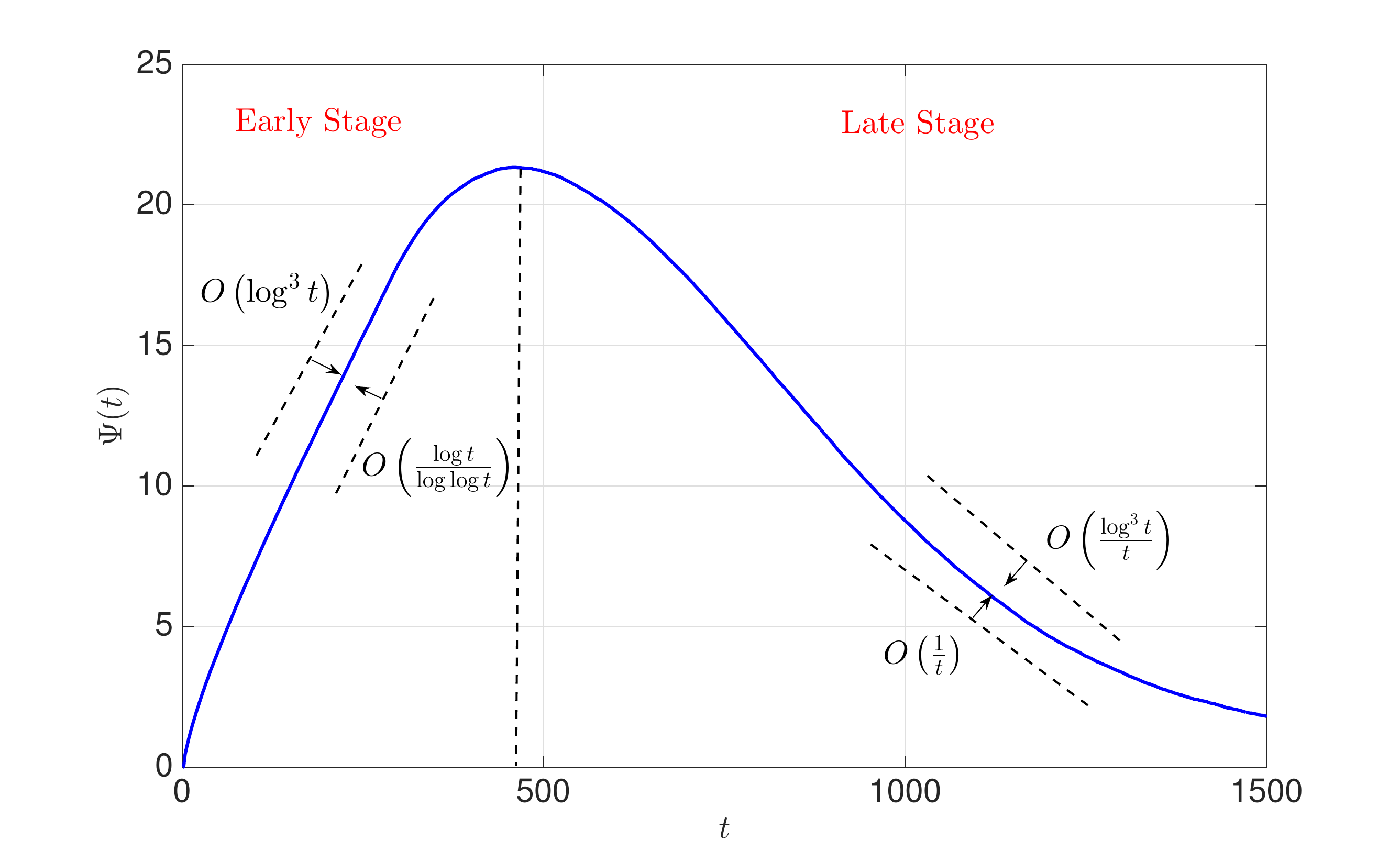}
\caption{Variation of queue-regret $\Psi(t)$ for a particular user under Q-UCB in a $1 \times 5$ system with $\epsilon = 0.15$ and $\Delta = 0.17$}
\label{fig:illustrate}
\end{center}
\end{figure}

As a preview of the theoretical results, Figure~\ref{fig:illustrate} shows the evolution of queue-regret with time in a single queue system with five servers under a scheduling policy inspired by UCB. (Further discussion of the scheduling algorithm used can be found in Section~\ref{subsec:alg}.) It is observed that the regret goes through a phase transition. In the initial stage, when the algorithm has not estimated the service probabilities well enough to stabilize the queue, the regret grows poly-logarithmically as in the classical MAB setting. After a critical point when the algorithm has learned the system parameters well enough to stabilize the queue, the queue-length goes through regenerative cycles as the queue becomes empty. Thus at the beginning of every regenerative cycle, there is no accumulation of past errors and the sample-path queue-regret is at most zero. As the algorithm estimates the parameters better with time, the length of the regenerative cycles decreases and the queue-regret decays to zero. 

An interesting question to ask is: how does the queue-regret scale as $t \rightarrow \infty$? For the basic MAB problem, it is well-known that regret, which is the cumulative sum of the rate loss in each time-slot, scales as $O(\log t)$. 
 However, queue-regret is the mean difference between the cumulative sum of departures. We show in the following lemma that, for any scheduling policy, queue-regret is upper bounded by the cumulative sum of rate loss, i.e., regret in the basic MAB problem. Let $\bm{S}(t)$ and $\bm{S}^*(t)$ denote the service offered by the servers scheduled at time $t$ by the proposed algorithm and the genie policy respectively. Note that, for any queue $u \in [U]$, $\EE\left[ S^*_u(t) - S_u(t) \right]$ is the rate loss at time $t$. Then, the queue-regret $\bm{\Psi}(t)$ has the following upper bound:
\begin{lemma}
\label{lem:easy-ub}
$\bm{\Psi}(t) = \EE\left[ \bm{Q}(t) - \bm{Q}^*(t) \right] \leq \sum_{l=1}^{t} \EE\left[ \bm{S}^*(l) - \bm{S}(l) \right].$
\end{lemma}
Please see Appendix~\ref{subsec:proof-easy-ub} for the proof for this lemma. This result shows that, in the single queue setting, using classical bandit algorithms like UCB1 and Thompson Sampling, one can achieve a queue-regret of at most $O(\log t)$.
However, our intuition suggests that the accumulation of errors is only over regenerative cycles.  Now observe that the derivative of cumulative regret (which roughly corresponds to regret per time-slot) is $O(1/t),$ and that the regenerative cycle-lengths for the optimal policy are $\Theta(1).$ Thus, it is reasonable to believe that the queue-regret at time $t$ is $O\left( 1/t \right)$ times a constant factor that increases with the length of the regenerative cycle.
\label{edit:forced-explore}
To push this intuition through to a formal proof requires high probability results for bandits over finite intervals of time corresponding to queue busy periods, where the intervals are random variables which are themselves coupled to the bandit strategy (the regenerative cycle-lengths are coupled to the bandit scheduling decisions). Analyses of traditional algorithms like UCB1 and Thompson Sampling guarantee high probability results only after sufficient number of sub-optimal arm pulls (\cite{auer2002finite}). \cite{audibert2009exploration} give upper tail bounds for the number of sub-optimal arm pulls for the UCB1 and UCB-V algorithms but these bounds are polynomial in the factor of deviation from the mean and not polynomial in time.

To the best of our knowledge, there is a lack of high probability upper and lower bounds in the multi-armed bandit literature for arbitrary time intervals.  The lack of such results motivates us to use alternate proof strategies for the lower bound on queue-regret and the achievability results. For the lower bound, we use coupling arguments to derive lower bounds on queue regret growth over one time-step for any $\alpha$-consistent policy. This allows us to show that no $\alpha$-consistent policy can achieve a better scaling than $O\left( 1/t \right).$  For the achievability result, we use a forced exploration variant of UCB/Thompson Sampling (a combination of $\epsilon$-greedy and UCB/Thompson Sampling algorithms) that gives us good upper bounds on the expected number of sub-optimal schedules over finite time intervals. When combined with queueing arguments, this leads to an upper bound on queue regret.  

As before, in each section, we first discuss all our main results in the context of a model with a single queue.  In particular, Section~\ref{ssec:lb-res} discusses our approach to the lower bound while Section~\ref{subsec:alg} discusses our upper bound.  We then extend the results in the single-queue setting to the more general model described in Section~\ref{subsec:sys-model-multiq}.

\subsection{An Asymptotic Lower Bound}
\label{ssec:lb-res}
We first establish a lower bound in a single-queue system; the arguments leading to this lower bound also form the backbone of our lower bound for the switch scheduling system.  

\subsubsection{The Single-Queue System.}  We establish an asymptotic lower bound on regret for the class of $\alpha$-consistent policies.  For the single-queue setting, we define this class of policies as in the traditional stochastic MAB problem (\cite{lai1985asymptotically, salomon2013lower, combes2015bandits}) as follows:
\begin{definition}
\label{defn:consistent}
A scheduling policy for a single queue network is said to be $\alpha$-consistent (for some $\alpha \in (0,1)$) if  for any problem instance $(\lambda,\pmb{\mu})$, there exists a constant   $C(\lambda,\pmb{\mu})$ such that
\begin{align*}
\mathbb{E} \left[ \sum _{s = 1} ^t\mathds{1}\lbrace\kappa(s) = k\rbrace\right] \leq   C(\lambda,\pmb{\mu})  t^{\alpha}
\end{align*}
for all $k \neq k ^*$.
\end{definition} 
This means that any policy under this class schedules a sub-optimal server not more than $O(t^{\alpha})$ (sub-linear in $t$) number of times.  Without a restriction such as that imposed in the preceding definition, ``trivial'' policies that, e.g., schedule the same server every time step would be allowed.  Such policies would have zero expected regret if the chosen server happened to be optimal, and otherwise would have linear regret.  The $\alpha$-consistency requirement rules out such policies, while ensuring the set under consideration is reasonable.
Proposition~\ref{thm:lb-late} below gives an asymptotic lower bound on the queue regret for an arbitrary $\alpha$-consistent policy.
\begin{proposition}
\label{thm:lb-late}
For any problem instance $(\lambda, \pmb{\mu})$ and any $\alpha$-consistent policy, the regret $\Psi(t)$  satisfies
\begin{align*}
\Psi(t) \geq \left( \frac{\lambda}{4}  D(\pmb{\mu}) (1 - \alpha) (K-1) \right) \frac{1}{t}
\end{align*}
for infinitely many $t$, where 
\begin{align}
\label{eq:constD}
D(\pmb{\mu}) := \frac{\Delta}{\mathrm{KL}\left(\mu_{min}, \frac{\mu^*+1}{2} \right)}.
\end{align}

\end{proposition}

\begin{proof}[Proof Outline for Proposition~\ref{thm:lb-late}.]
The proof of the lower bound consists of three main steps. First, in Lemma~\ref{lem:lb-queue1}, we show that the regret at any time-slot is lower bounded by the probability of a sub-optimal schedule in that time-slot (up to a constant factor that is dependent on the problem instance). The key idea in this lemma is to show,  for any scheduling algorithm, the equivalence of two systems with the same marginal service distributions. This is achieved through a carefully constructed coupling argument that maps the original system with independent service across links to another system with service process that is dependent across links but with the same marginal distribution.

 As a second step, the lower bound on the regret in terms of the probability of a sub-optimal schedule enables us to obtain a lower bound on the cumulative queue-regret in terms of the number of sub-optimal schedules. We then use a lower bound on the number of sub-optimal schedules for $\alpha$-consistent policies (Lemma~\ref{lem:banditlb} and Corollary~\ref{cor:lem:banditlb}) to obtain a lower bound on the cumulative regret. In the final step, we use the lower bound on the cumulative queue-regret to obtain an {\it infinitely often} lower bound on the queue-regret.  
\end{proof}
We now generalize this result to the switch network (Theorem~\ref{thm:lb-late-multiq}). We omit the proof details of Proposition~\ref{thm:lb-late} and give a detailed proof for Theorem~\ref{thm:lb-late-multiq} in Appendix~\ref{sec:proofs-lb}.
\subsubsection{The Multi-Queue Network.}
Before stating the generalization of Proposition~\ref{thm:lb-late} for the switch network, we first extend the definition of the class of $\alpha$-consistent policies for a multi-queue network with a unique optimal matching.
\begin{definition}
\label{defn:consistent-multiq}
A scheduling policy for a multi-queue network with a unique optimal matching is said to be $\alpha$-consistent (for some $\alpha \in (0,1)$) if  for any problem instance $(\pmb{\lambda},\pmb{\mu})$, there exists a constant   $C(\pmb{\lambda},\pmb{\mu})$ such that
\begin{align*}
\mathbb{E} \left[ \sum _{s = 1} ^t\mathds{1}\lbrace\kappa _u(s) = k\rbrace\right] \leq  C(\pmb{\lambda},\pmb{\mu}) t^{\alpha}
\end{align*}
for all $u \in [U]$ and $k \neq k ^*_u$.
\end{definition} 

Theorem~\ref{thm:lb-late-multiq} below gives an asymptotic lower bound on the average queue-regret and per-queue regret for an arbitrary $\alpha$-consistent policy in the multi-queue setting; it has the same scaling in $t$ as in the single-queue setting, but now explicitly depends on the parameters of the switch problem.
\begin{theorem}
\label{thm:lb-late-multiq}
For any problem instance $(\pmb{\lambda}, \pmb{\mu})$ with a unique optimal matching, and any $\alpha$-consistent policy, the regret $\pmb{\Psi}(t)$  satisfies	
\begin{enumerate}[label=(\alph*)]
\item \label{item:lb-late-avg} 
\begin{align*}
\frac{1}{U} \sum_{u \in [U]} \Psi_u(t) \geq \left( \frac{\lambda_{min}}{8}  D(\pmb{\mu}) (1 - \alpha) (K-1) \right) \frac{1}{t},
\end{align*}
\item \label{item:lb-late-single} and for any $u \in [U]$, 
\begin{align*}
\Psi_u(t) \geq \left( \frac{\lambda_{min}}{8}  D(\pmb{\mu}) (1 - \alpha) \max \left\lbrace U-1, 2(K-U) \right\rbrace \right) \frac{1}{t}
\end{align*}
\end{enumerate}
for infinitely many $t$, where $D$ is given by \eqref{eq:constD}.
\end{theorem}
 This result does not require \cref{ass:initial-state} and holds irrespective of the distribution of the initial queue-length. The full proof of this theorem is given in Section~\ref{sec:proofs-lb}.

\subsection{Achieving the Asymptotic Bound}
\label{subsec:alg}
The lower bounds in Proposition~\ref{thm:lb-late} and Theorem~\ref{thm:lb-late-multiq} imply that no $\alpha$-consistent policy can achieve a queue-regret better than $O(1/t).$   We next ask if straight-forward generalizations of standard bandit algorithms like UCB and Thompson sampling can achieve a scaling of $O\left( 1/t \right),$ thus matching the lower bound in Theorem~\ref{thm:lb-late-multiq}. To prove that these algorithms achieve this scaling, it is essential to show high probability bounds on scheduling errors over regenerative cycles in the late stage. A systematic way to show this would be to prove that the algorithm has a good estimate of all the service probabilities in the late stage leading to the correct scheduling decision. But for standard bandit algorithms, a lack of concentration results on the number of times each link is scheduled makes it difficult to prove a high probability bound on scheduling errors over a finite time-interval in the late stage.

\subsubsection*{Algorithms with Forced Exploration.}
To get around this difficulty, we propose slightly modified versions of the standard bandit algorithms (UCB1 and Thompson Sampling) generalized to the multi-dimensional queueing bandit. These algorithms, which we call Q-UCB and Q-ThS (corresponding to UCB1 and Thompson sampling, respectively) have an explicit forced exploration component similar to $\epsilon$-greedy algorithms.  This forced exploration may not be necessary to ensure good performance in practice (see \cref{sec:simulations} for a detailed discussion) but it helps us in theoretically proving that the algorithm has a sufficiently good estimate of all the service probabilities (including sub-optimal ones) in the late stage.

We first briefly describe the algorithms in the single queue setting.  The algorithms are precisely defined in a more general setting with multiple queues in Section~\ref{subsubsec:alg-multiq}.  In particular, notation and details of the algorithms are given for the general case in Table~\ref{tab:notation-alg} and Algorithms~\ref{alg:ucb} and \ref{alg:ths} respectively.

In the single queue setting, at time-slot $t$ ($t \geq 1$), both Q-UCB  and Q-ThS  \emph{explore} with probability $\min\{1,3K \log^2t/t\},$ otherwise they \emph{exploit}. To explore, the algorithm chooses a server uniformly at random from the $K$ servers.  The chosen exploration rate ensures that we are able to obtain concentration results for the number of times any link is sampled.\endnote{The exploration rate could scale like $\log t/t$ if we knew $\Delta$ in advance; however, without this knowledge, additional exploration is needed.}  If it exploits, Q-UCB (resp., Q-ThS) chooses the best server using the UCB1 (resp., Thompson sampling) method. In other words, Q-UCB chooses the server that has the highest upper confidence bound for the corresponding service probability, while Q-ThS chooses the server according to random samples drawn from the posterior distribution of the service probabilities.  Note that both UCB1 and Thompson sampling explore as well; the innovation in our algorithms is the {\em forced} exploration that we have added, to provide high confidence estimation of service probabilities  (please see \cref{lem:exploit-ub} in Appendix~\ref{sec:proofs-ub} for details).

\subsubsection*{Asymptotic Guarantees for Q-UCB and Q-ThS.}
For a given problem instance $(\pmb{\lambda}, \pmb{\mu})$  (and therefore fixed $\pmb{\epsilon}$), we can show that the regret under Q-UCB and Q-ThS scale as $O\left( \mathrm{poly}(\log t)/t \right)$.
We state our asymptotic upper bound in the case of a single queue in Proposition~\ref{cor:thm:ucb-best-match}; this is a special case of Corollary~\ref{cor:thm:ucb-best-match-multiq} in Section \ref{subsubsec:alg-multiq}, which holds for a more general setting with multiple queues.


\begin{proposition}
\label{cor:thm:ucb-best-match}
Let   $w(t) = t^{(1 - 1/\beta)}$ for some fixed $\beta > 1$. Then for both Q-UCB and Q-ThS,
$$\Psi(t) = O\left(K \frac{ \log^3 t}{\epsilon^2 t} \right)$$ for all $t$ such that
$\frac{ w(t)}{\log t} \geq \frac{2}{\epsilon}$, $\frac{t}{ w(t)} \geq \max\left\lbrace  \frac{24 K}{\epsilon}, 15 K^2 \log t \right\rbrace$,  $t \geq \exp \left( \frac{4}{\Delta^2 (1 - 1/\beta)^3} \right)$ and $\frac{t}{\log t} \geq \frac{198}{\epsilon^2}$.
\end{proposition}

\begin{proof}[Proof Outline for Proposition~\ref{cor:thm:ucb-best-match}.]
We outline how the proof proceeds here for a single queue, to lay bare the central arguments.  The full proof for Theorem \ref{thm:ucb-best-match-multiq}, which is a generalization of Proposition~\ref{cor:thm:ucb-best-match} for the setting of multiple queues, is provided in Section~\ref{sec:proofs-ub}.

As mentioned earlier, the central idea in the proof is that the sample-path queue-regret is at most zero at the beginning of regenerative cycles, i.e., instants at which the queue becomes empty. The proof consists of two main parts -- one which gives a high probability result on the number of sub-optimal schedules in the exploit phase in the late stage, and the other which shows that at any time, the beginning of the current regenerative cycle is not very far in time. 

The former part is proved in Lemma~\ref{lem:exploit-ub}, where we make use of the forced exploration component of Q-UCB and Q-ThS to show that  all the links, including the sub-optimal ones, are sampled a sufficiently large number of times to give a good estimate of the service probabilities. This in turn ensures that the algorithm schedules the correct links in the exploit phase in the late stages with high probability.
 
For the latter part, we prove a high probability bound on the last time instant when the queue was zero (which is the beginning of the current regenerative cycle) in Lemma~\ref{lem:busy-period-ub2}. Here, we make use of a recursive argument to obtain a tight bound. More specifically, we first use a coarse high probability upper bound on the queue-length (Lemma~\ref{lem:ub-queue1}) to get a first cut bound on the beginning of the regenerative cycle (Lemma~\ref{lem:busy-period-ub1}). This bound on the regenerative cycle-length is then recursively used to obtain tighter bounds on the queue-length, and in turn, the start of the current regenerative cycle (Lemmas~\ref{lem:ub-queue2} and \ref{lem:busy-period-ub2} respectively).

The proof proceeds by combining the two parts above to show that the main contribution to the queue-regret comes from the forced exploration component in the current regenerative cycle, which gives the stated result.  
\end{proof}

 This result, in combination with Proposition~\ref{thm:lb-late}, shows that queue-regret for Q-UCB and Q-ThS in the long-term is within a $\mathrm{poly}(\log t)$ factor of the optimal queue-regret for the $\alpha$-consistent class. 
  We will next describe how this result can be extended to the multi-queue switch network. We skip the details of the proof for Proposition~\ref{cor:thm:ucb-best-match} and give a detailed proof for its generalized version (Corollary~\ref{cor:thm:ucb-best-match-multiq})  in Appendix~\ref{sec:proofs-ub}.

\subsubsection{The Multi-Queue Network.}
\label{subsubsec:alg-multiq}
We now describe the scheduling algorithm in the multi-queue setting.  Let $\mathcal{M} \subset [K]^U$ be the set of all matchings.  We represent a matching by a $U$-length vector in which the $u^{th}$ element gives the server that is assigned to queue $u$. A vector $\pmb{\kappa} \in \mathcal{M}$ if and only if for any $u' \neq u,$ $\kappa_{u'} \neq \kappa_u$. Let $\mathcal{X} \subset \mathcal{M}$ be a subset of $K$ perfect matchings such that their union covers the set of all edges in the complete bipartite graph (it is easy to show that such a decomposition is possible). Also, let $T_{uk}(t)$ be the number of times server $k$ is assigned to queue $u$ in the first $t$ time-slots and $\hat{ \pmb{\mu} }(t)$ be  the empirical mean of the service offered by the links at time $t$ from past observations (until $t-1$). 

At time-slot $t$, Q-UCB and Q-ThS both decide to \emph{explore} with probability $\min\{1,3K \log^2t/t\},$ otherwise they \emph{exploit}. To explore, the algorithms choose a matching uniformly at random from the set $\mathcal{X}.$  If it exploits, Q-UCB (resp., Q-ThS) first assigns a best server for every queue using the UCB1 (resp., Thompson Sampling) method. In other words, for every queue, Q-UCB chooses the server that has the highest upper confidence bound for the corresponding service probability (breaking ties arbitrarily), while Q-ThS chooses the server according to random samples drawn from the posterior distribution of the service probabilities. Q-UCB (resp., Q-ThS) then schedules the projection of this $U$-length server vector $\pmb{\hat{k}}(t)$  onto the space of all matchings $\mathcal{M}$. Notation and details of the algorithms are given in Table~\ref{tab:notation-alg} and Algorithms~\ref{alg:ucb} and \ref{alg:ths}.

Note that the scheduled matching $\pmb{\kappa}(t)$ is the projection of $\pmb{\hat{k}}(t)$ onto the space of all matchings $\mathcal{M}$ with Hamming distance as metric,  i.e., a matching that matches the maximum number of queues with their corresponding ``best'' server.  One way to compute this projection is to take the union of two matchings selected as follows:
\begin{enumerate}
\item The first is a maximal matching in the sub-graph induced by the best server configuration;
\item The second is a maximal matching in the sub-graph obtained by removing the matching chosen in the first step from the complete bipartite graph.
\end{enumerate}

\begin{table}
\renewcommand{\arraystretch}{1.3}
\caption{Notation for Algorithms~\ref{alg:ucb}, \ref{alg:ths}}
\label{tab:notation-alg}
\centering
\begin{tabular}{||c | c||}
\hline
\bfseries Symbol & \bfseries Description	\\
\hline\hline
$\mathsf{E}(t)$ & Indicates if the algorithm schedules a matching through \textit{Explore}	\\ \hline	
$\mathsf{E}_{uk}(t)$ &Indicates if Server $k$ is assigned to Queue $u$ at time $t$ through \textit{Explore}	\\ \hline	
$\mathsf{I}_{uk}(t)$ & Indicates if Server $k$ is assigned to Queue $u$ at time $t$ through \textit{Exploit}	\\ \hline
$T_{uk}(t)$ & Number of time slots Server $k$ is assigned to Queue $u$ in time $[1,t-1]$	\\ \hline
$\hat{ \bm{\mu} }(t)$ & Empirical mean of service at time $t$ from past observations (until $t-1$)	\\ \hline
$\pmb{\kappa}(t)$ & Matching scheduled in time-slot $t$	\\ \hline
\hline
\end{tabular}
\end{table}

\begin{algorithm}
  \caption{Q-UCB}
  \begin{algorithmic}
    \State At time $t \geq 1$, 
    \State Let $\mathsf{E}(t)$ be an independent Bernoulli sample of mean $\min\{1,3K\frac{\log^2t}{t}\}.$
    \If { $\mathsf{E}(t) = 1$}
    \State \textit{Explore:} 
    \State Schedule a matching from $\mathcal{X}$ uniformly at random.
    \Else
    \State \textit{Exploit:}
    \State Compute for all $u \in [U]$
    \begin{align*}
    \hat{k}_u(t) & : = \arg\max_{k \in [K]} \hat{\mu}_{uk}(t) + \sqrt{\frac{\log^2 t}{2 T_{uk}(t-1)}}.
    \end{align*}
    \State Schedule a matching $\pmb{\kappa}(t)$ such that
    \begin{align*}
    \pmb{\kappa}(t) & \in \arg\min_{\pmb{\kappa} \in \mathcal{M}} \sum_{u \in [U]} \mathds{1}\left\lbrace \kappa_u \neq \hat{k}_u(t) \right\rbrace,
    \end{align*}   
    \EndIf 
  \end{algorithmic}
    \label{alg:ucb}
\end{algorithm}
 
\begin{algorithm}
  \caption{Q-ThS}
  \begin{algorithmic}
    \State At time $t \geq 1$, 
    \State Let $\mathsf{E}(t)$ be an independent Bernoulli sample of mean $\min\{1,3K\frac{\log^2t}{t}\}.$
    \If { $\mathsf{E}(t) = 1$}
    \State \textit{Explore:} 
    \State Schedule a matching from $\mathcal{X}$ uniformly at random.
    \Else
    \State \textit{Exploit:}
    \State For each $k \in [K], u \in [U]$ , pick a sample $\hat{\theta}_{uk}(t)$ of distribution,
    \begin{align*}
    \hat{\theta}_{uk}(t) \sim \mathrm{Beta}\left(\hat{\mu}_{uk}(t) T_{uk}(t-1) + 1, \left( 1 - \hat{\mu}_{uk}(t) \right) T_{uk}(t-1) + 1 \right).
\end{align*}   
    \State Compute for all $u \in [U]$
    \begin{align*}
    \hat{k}_u(t) & : = \arg\max_{k \in [K]} \hat{\theta}_{uk}(t)
    \end{align*}
    \State Schedule a matching $\pmb{\kappa}(t)$ such that
    \begin{align*}
    \pmb{\kappa}(t) & \in \arg\min_{\pmb{\kappa} \in \mathcal{M}} \sum_{u \in [U]} \mathds{1}\left\lbrace \kappa_u \neq \hat{k}_u(t) \right\rbrace,
    \end{align*} 
    \EndIf 
  \end{algorithmic}
    \label{alg:ths}
\end{algorithm}

Let $\tau_1 = 5.8 \times 10^3$, and let $\tau_2$ be a constant such that 
\begin{align}
\label{eq:large-const-tau}
t \exp\left( -\frac{1}{2} \log^2 t \right) +  t^2 \exp \left( - \frac{1}{4} \left( 2 \log t \right)^{4/3} \right) +  t^2 \exp \left( - \frac{1}{4} \log^2 t \right) \leq \frac{1}{6t^3}
\end{align}
for all $t \geq \tau_2.$ Such a $\tau_2$ exists since each term on the left-hand side of \eqref{eq:large-const-tau} is $o(1/t^3)$.

The following theorem gives an upper bound on the regret for each individual queue for both Q-UCB (Algorithm~\ref{alg:ucb}) and Q-ThS (Algorithm~\ref{alg:ths}) in the multi-queue setting. 
\begin{theorem}
\label{thm:ucb-best-match-multiq}
Consider any problem instance $(\pmb{\lambda}, \pmb{\mu})$ which has a unique optimal matching. For any $u \in [U]$, let   $w(t) = t^{(1 - 1/\beta)}$ for some fixed $\beta > 1$, $v'_u(t) =  \frac{6K}{\epsilon_u} w(t)$ and $v_u(t) = \frac{24}{\epsilon_u^2} \log t + \frac{60 K}{\epsilon_u} \frac{v'_u(t) \log^2 t}{t}.$
Then, for Algorithm~\ref{alg:ucb} (resp., Algorithm~\ref{alg:ths}), the regret for queue $u$, $\Psi_u(t)$, satisfies 
 $$\Psi_u(t) \leq 6K\frac{v_u(t) \log^2 t}{t} + \frac{24.004 + UK}{6 t^2}$$ for all $t \geq \tau_1$ (resp., $t \geq \tau_2$) such that $t \geq  \exp \left( \frac{4}{\Delta^2 (1 - 1/\beta)^3} \right)$, $\frac{ w(t)}{\log t} \geq \frac{2}{ \epsilon_u}$ and $v_u(t) + v'_u(t) \leq t/2$.
\end{theorem}
 
\begin{remark}
In our analysis, the constants $\tau_1$ and $\tau_2$ are determined by the upper bounds obtained in Lemmas~\ref{lem:explore-ub} and \ref{lem:exploit-ub}. For the upper bound on the number of explore time-slots in Lemma~\ref{lem:explore-ub}, we show that is sufficient to have $t \geq  5.8 \times 10^3$ for both Q-UCB and Q-ThS. For the upper bound on the number of sub-optimal schedules in the exploit phase, our analysis in Lemma~\ref{lem:exploit-ub} shows that it is sufficient to take $\tau_1 = 5.8 \times 10^3$ for Q-UCB. For Q-ThS, we use the `Beta-Binomial trick' from \cite{kaufmann2012thompson,agrawal-goyal12thompson} to get an upper bound which is qualitatively similar to that for Q-UCB (please see the bounds in \eqref{eq:subopt-prob-ucb} and \eqref{eq:subopt-prob-ths} for comparison). However, the constants in the exponent for Q-ThS are smaller than that of Q-UCB by a factor of 4 making $\tau_2$ a much larger constant than $\tau_1$.  While the upper bounds in this paper are suggestive of the evolution of regret with time, the obtained constants do not accurately characterize the empirical performance of the proposed algorithms. In particular, the inferior constants for Q-ThS as compared to Q-UCB could just be an artefact of our analysis. As seen in Figure~\ref{fig:compare}, simulations show that Q-ThS performs better than Q-UCB both in the early and late stages.
\end{remark}

\begin{remark}
For any queue $u$, we state the regret bounds and the corresponding time-intervals in which these bounds hold as a function of $\epsilon_u$, the gap between the arrival probability and the best service probability for that queue. Therefore, the time ranges for which the bounds hold may vary for different queues depending on $\pmb{\epsilon}$.
\end{remark}
\begin{remark}
Although we assume Bernoulli distributions for arrival and service in our model, the result in Theorem~\ref{thm:ucb-best-match-multiq} holds for general, non-Bernoulli distributions with bounded support if
\begin{enumerate}[label=(\roman*)]
\item there is a unique matching that gives the best service rate for all users (similar to Assumption~\ref{ass:best-match} in Section~\ref{sec:sys-model}), and
\item the genie policy that defines the regret $\pmb{\Psi}(t)$ is the one that always schedules the best matching.
\end{enumerate}
\end{remark}

\begin{remark}
The result in Theorem~\ref{thm:ucb-best-match-multiq} can be proved without making any assumption on the distribution of $\mathbf{Q}^*(0)$ as in \cref{ass:initial-state}. Specifically, this assumption is used in \cref{lem:ub-queue1} to  show a first cut upper bound on the queue-length. Using the fact that the process  $\{\mathbf{Q}^*(t)\}$ is geometrically ergodic, the same lemma can be extended to the case of a generic $\mathbf{Q}^*(0)$ but the convergence time would now be a function of $\mathbf{Q}^*(0)$.
\end{remark}

 A slightly weaker version of Theorem~\ref{thm:ucb-best-match-multiq} is given in Corollary~\ref{cor:thm:ucb-best-match-multiq}. This corollary is useful to understand the dependence of the upper bound on the load $\pmb{\epsilon}$ and the number of servers $K$. (Proposition~\ref{cor:thm:ucb-best-match} is a special case of this corollary.)
\begin{corollary}
\label{cor:thm:ucb-best-match-multiq}
For any $\beta > 1$, $$\Psi_u(t) \leq \frac{289 K \log^3 t }{\epsilon_u^2 t}$$ for all $t \geq \tau_0$ such that
$$\log t \geq \max \left\lbrace \frac{4}{\Delta^2 (1 - 1/\beta)^3}, \sqrt{2} \left( \log \frac{2}{\epsilon_u} \right)^{1.5},  \beta \frac{24 K}{\epsilon_u}, \beta \left( \log \log t + \log(15 K^2) \right), \frac{\beta}{\beta - 1} \log\left( \frac{13.2}{K^2 \epsilon_u^2} \right)  \right\rbrace.$$
\end{corollary}
Full proofs of Theorem~\ref{thm:ucb-best-match-multiq} and Corollary~\ref{cor:thm:ucb-best-match-multiq} are given in Section~\ref{sec:proofs-ub}.

%% file: early-stage.tex
\section{The Early Stage in the Heavily Loaded Regime}
\label{sec:early-stg}

Proposition~\ref{cor:thm:ucb-best-match}, Theorem~\ref{thm:ucb-best-match-multiq} show that systematic exploration in Q-UCB and Q-ThS ensures an $O\left(\mathrm{poly}(\log t)/t \right)$ queue-regret in the late stage. The penalty for aggressive exploration is likely to be more apparent in the initial stages when the queues have not yet stabilized and there are few regenerative cycles. As a result, the queueing system has a behavior similar to the traditional MAB system in the early stage. Thus, it is reasonable to expect that algorithms that achieve optimal performance for the traditional MAB problem also perform well in the early stages in the queueing system.

In order to study the performance of $\alpha$-consistent policies in the early stage, we again focus on a single queue, and consider the \emph{heavily loaded} system, where the arrival probability $\lambda$ is close to the optimal service probability $\mu^*$. Specifically, we characterize the behavior of queue-regret as the difference between the two probabilities, $\epsilon = \mu^* - \lambda \rightarrow 0$.  As in the discussion of the late stage, we present all our results first within this single queue model to highlight the main insights, but these results are in fact special cases of results that are proven in the context of a general switch network setting.

Analyzing regret in the early stage in the heavily loaded regime has the effect that the optimal server is the only one that stabilizes the queue.  As a result, in the heavily loaded regime, effective learning and scheduling of the optimal server play a crucial role in determining the transition point from the early stage to the late stage.  For this reason the heavily loaded regime reveals the behavior of regret in the early stage.
%

 Proposition~\ref{thm:lb-early} gives a lower bound on the regret in the heavily loaded regime, roughly in the time interval $\left(  K^{1/1 - \alpha }, O \left( K/\epsilon \right) \right)$ for any $\alpha$-consistent policy.  Recall the definition of $D(\pmb{\mu})$ given by Equation~\eqref{eq:constD}.

\begin{proposition}
\label{thm:lb-early}
Given any single queue system $(\lambda,\pmb{\mu})$, and for any $\alpha$-consistent policy and $\gamma  > \frac{1}{1 - \alpha}$, there exist constants $\tau$ and \const{const:switch-lb}$C_{\ref{const:switch-lb}}$ (independent of $(\lambda,\pmb{\mu})$) such that, for $\eta := \frac{(K-1)D(\pmb{\mu})}{2\max \{C_{\ref{const:switch-lb}} K^{\gamma }, \tau \}}$, if $\epsilon < \eta$, then the regret $\Psi(t)$  satisfies		
 \begin{align*}
\Psi(t) \geq \frac{D(\pmb{\mu})}{2} (K-1) \frac{\log t}{\log \log t}
\end{align*}
for $t \in \left[ \max \{C_{\ref{const:switch-lb}} K^{\gamma }, \tau \}, (K-1)\frac{D(\pmb{\mu})}{2\epsilon} \right]$.
\end{proposition}

\begin{proof}[Proof Outline for Proposition~\ref{thm:lb-early}.]
The crucial idea in the proof is to show a lower bound on the queue-regret in terms of the number of sub-optimal schedules (Lemma~\ref{lem:lb-queue2}). As in Theorem~\ref{thm:lb-late}, we then use a lower bound on the number of sub-optimal schedules for $\alpha$-consistent policies (given by Corollary~\ref{cor:lem:banditlb}) to obtain a lower bound on the queue-regret.  
\end{proof}

Proposition~\ref{thm:lb-early} is a special case of Theorem~\ref{thm:lb-early-multiq} (given below) in the multi-queue setting. It gives lower bounds on the average queue regret and individual queue regret in the early stage.  \cref{ass:initial-state} is useful in this result to ensure that the initial queue-length is not far-off from the optimal (see \cref{lem:lb-queue2}). This assumption can be weakened to the following requirement: ``the expected initial queue-length is bounded above by the expected stationary queue-length of $\mathbf{Q}^*(t)$''. The proof of this theorem is given in Section~\ref{sec:proofs-lb}. Recall that $\epsilon_u = \mu_u^* - \lambda_u$ for $u \in [U]$.
 
\begin{theorem}
\label{thm:lb-early-multiq}
Given any problem instance $(\pmb{\lambda},\pmb{\mu})$, and for any $\alpha$-consistent policy and $\gamma  > \frac{1}{1 - \alpha}$, there exist constants $\tau$ and  $C_{\ref{const:switch-lb}}$ (independent of $(\lambda,\pmb{\mu})$) such that, for $\eta := \frac{(K-1)D(\pmb{\mu})}{2\max \{C_{\ref{const:switch-lb}} K^{\gamma }, \tau \}}$, the following holds:
\begin{enumerate}[label=(\alph*)]
\item \label{item:lb-early-avg}
if $\bar{\epsilon} := \frac{1}{U} \sum_{u \in [U]} \epsilon_u < \frac{\eta}{2}$, then the average regret satisfies		 
\begin{align*}
\frac{1}{U} \sum_{u \in [U]} \Psi _{u} (t) \geq  \frac{D(\pmb{\mu})}{4} (K-1) \frac{\log t}{\log \log t},
\end{align*}
for  $t \in \left[ \max \{C_{\ref{const:switch-lb}} K^{\gamma }, \tau \}, (K-1)\frac{D(\pmb{\mu})}{4\bar{\epsilon}} \right]$, and
\item \label{item:lb-early-single} for any $u \in [U]$, if $\epsilon_u < \eta$, then the regret for queue $u$ satisfies
\begin{align*}
\Psi_u(t) \geq \frac{D(\pmb{\mu})}{4}  \max \left\lbrace U-1, 2(K-U) \right\rbrace \frac{\log t}{\log \log t}
\end{align*}
for $t \in \left[ \max \{C_{\ref{const:switch-lb}} K^{\gamma }, \tau \}, (K-1)\frac{D(\pmb{\mu})}{2\epsilon_u} \right]$.	
\end{enumerate}
\end{theorem} 

Recall that the lower bound in Proposition~\ref{thm:lb-early} and Theorem~\ref{thm:lb-early-multiq}   aim to characterize
  the behavior of regret in the heavily loaded regime, i.e., as  $\epsilon := \mu^* - \lambda \to 0$. While the left end point of
   the interval in which the bound holds depends on the policy, the  right end point of the interval is policy independent.
   Further, for a fixed policy, the left endpoint does not change  as $\epsilon \to 0$ (i.e., it is independent of $\epsilon$),
   while the right endpoint grows without bound as $\epsilon \to 0$.  Therefore, for any $\alpha$-consistent
   policy, there is an $\epsilon$ small enough such that, for a  non-empty time-interval (until
   $(K-1)\frac{D(\pmb{\mu})}{2\epsilon}$), the regret grows at least as $\Omega\left( K \frac{\log t}{\log \log t}\right)$.\footnote{Note that the interval in Proposition~\ref{thm:lb-early} is non-empty for $\epsilon$ satisfying the condition of the proposition; the same is analogously true for $\epsilon_u$ satisfying the conditions in Theorem~\ref{thm:lb-early-multiq}.}  The main
   insight from this result is that it takes at least
   $\Omega \left( \frac{K}{\epsilon} \right)$ time for the queue-regret with any policy to transition from the early stage to the late stage.

Theorem~\ref{thm:lb-early-multiq} shows that, for any $\alpha$-consistent policy, it takes at least $\Omega \left( K/\epsilon \right)$ time for the queue-regret to transition from the early stage to the late stage.

 In this region, regret is growing with time, and the scaling $\Omega(\log t/\log \log t)$ reflects the fact that in this regime queue-regret is dominated by the fact that cumulative regret grows like $\Omega(\log t)$.  A reasonable question then arises: after time $\Omega \left( K/\epsilon \right)$, should we expect the regret to transition into the late stage regime analyzed in the preceding section?

We answer this question by studying when Q-UCB and Q-ThS achieve their late-stage regret scaling of $O \left( \mathrm{poly}(\log t)/\epsilon^2 t\right)$. 
 Formally, we have Corollary~\ref{cor:cor:thm:ucb-best-match-multiq}, which helps in understanding the scaling of the upper bound in Theorem~\ref{thm:ucb-best-match-multiq} with respect to the parameters $K$ and $\epsilon$.
\begin{corollary}
\label{cor:cor:thm:ucb-best-match-multiq}
For any problem instance $(\pmb{\lambda}, \pmb{\mu})$, any queue $u \in [U]$, any $\beta < 2$, $\gamma > \frac{\beta}{\beta - 1}$, $\delta > \beta$, there exists a constant \const{const:ucb-time-lb}$C_{\ref{const:ucb-time-lb}}$ independent of $K$ and $\epsilon_u$  (but depending on $\Delta,$ $\beta$, $\gamma$ and $\delta$) such that the regret for queue $u$ satisfies 
$$\Psi_u(t) \leq \frac{289 K \log^3 t}{\epsilon_u^2 t}$$ 
$\forall t \geq C_{\ref{const:ucb-time-lb}} \max \left\lbrace \left( \frac{1}{\epsilon_u} \right)^{\gamma},  \left( \frac{K}{\epsilon_u} \right)^{\beta} ,  K^{2\delta} \right\rbrace$.
\end{corollary}

As an illustrative example, we consider the case where $K = 1/\epsilon_u$ to compare the scaling of the upper bound with the lower bound in Theorem~\ref{thm:lb-early-multiq}. 
In this case, for any $\delta> 1.5$, there is a constant $C(\Delta, \delta)$ such that the time taken to achieve $O \left( K \mathrm{poly}(\log t)/\epsilon^2 t\right)$ regret scaling is $C(\Delta, \delta) \left( K/\epsilon \right)^{\delta}$. On the other hand, under any $\alpha$-consistent policy, the regret scales as $\Omega \left( K \log t/\log \log t\right)$ for $t < K/\epsilon$.

We conclude by noting that 
Q-UCB and Q-ThS do not yield optimal regret performance in the early stage in general.  In particular, recall that at any time $t$, the forced exploration component in Q-UCB and Q-ThS is invoked with probability $3K\log^2 t/t$. As a result, we see that, in the early stage, queue-regret under Q-UCB and Q-ThS could be a $\log^2 t$-factor worse than the $\Omega \left( \log t/\log \log t\right)$ lower bound shown in Theorem~\ref{thm:lb-early-multiq} for the $\alpha$-consistent class.  This intuition can be formalized: it is straightforward to show an upper bound of $2K \log^3 t$ for any \const{const:ucb-trivial-ub}$t > \max \{C_{\ref{const:ucb-trivial-ub}}, U\}$, where $C_{\ref{const:ucb-trivial-ub}}$ is a constant that depends on $\Delta$ but is independent of $K$ and $\epsilon$; we omit the details. 

%% file: simulationsv2.tex
\section{Simulation Results}
\label{sec:simulations}
In this section, we present simulation results to empirically evaluate the performance of the algorithms presented in this paper. These results corroborate our theoretical analysis in Sections~\ref{sec:late-stg} and \ref{sec:early-stg}. In particular a phase transition from unstable to stable behavior can be observed in all our simulations, as predicted by our analysis. Below, we demonstrate the performance of Q-ThS (Algorithm~\ref{alg:ths}) under variations of system parameters like the traffic ($\epsilon$), the gap between the optimal and the suboptimal servers ($\Delta$), and the size of the system ($U$ and $K$). We also compare the performance of our algorithms with the traditional bandit algorithms -- UCB1 (\cite{auer2002finite}) and Thompson Sampling (\cite{thompson1933likelihood}), which do not incorporate forced exploration.   

\subsubsection*{Variation with $U,K$ and $\pmb{\epsilon}$.}
\begin{figure*}
\begin{center}
\subfloat[Queue-Regret under Q-ThS for systems with $1$ queue and $5,7$ servers with $\epsilon \in \{ 0.05,0.1,0.15 \}$ \label{fig:kvar1}]{\includegraphics[width=.45\linewidth]{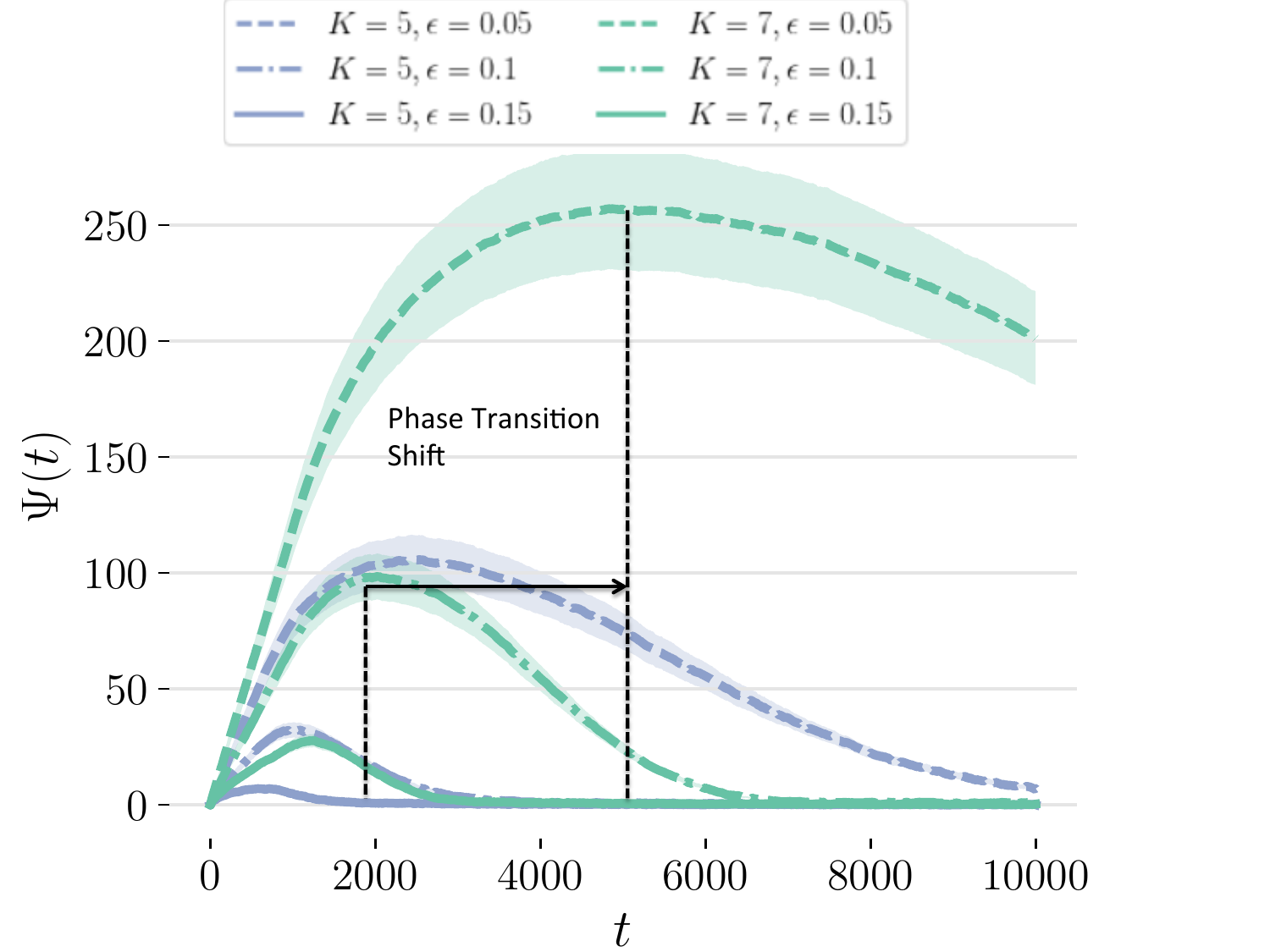}}
\hfill
\subfloat[Queue-Regret under Q-ThS for systems with $1,3$ queues and $5$ servers with $\epsilon  \in \{ 0.05,0.1,0.15 \}$ \label{fig:kvar2}]{\includegraphics[width=.45\linewidth]{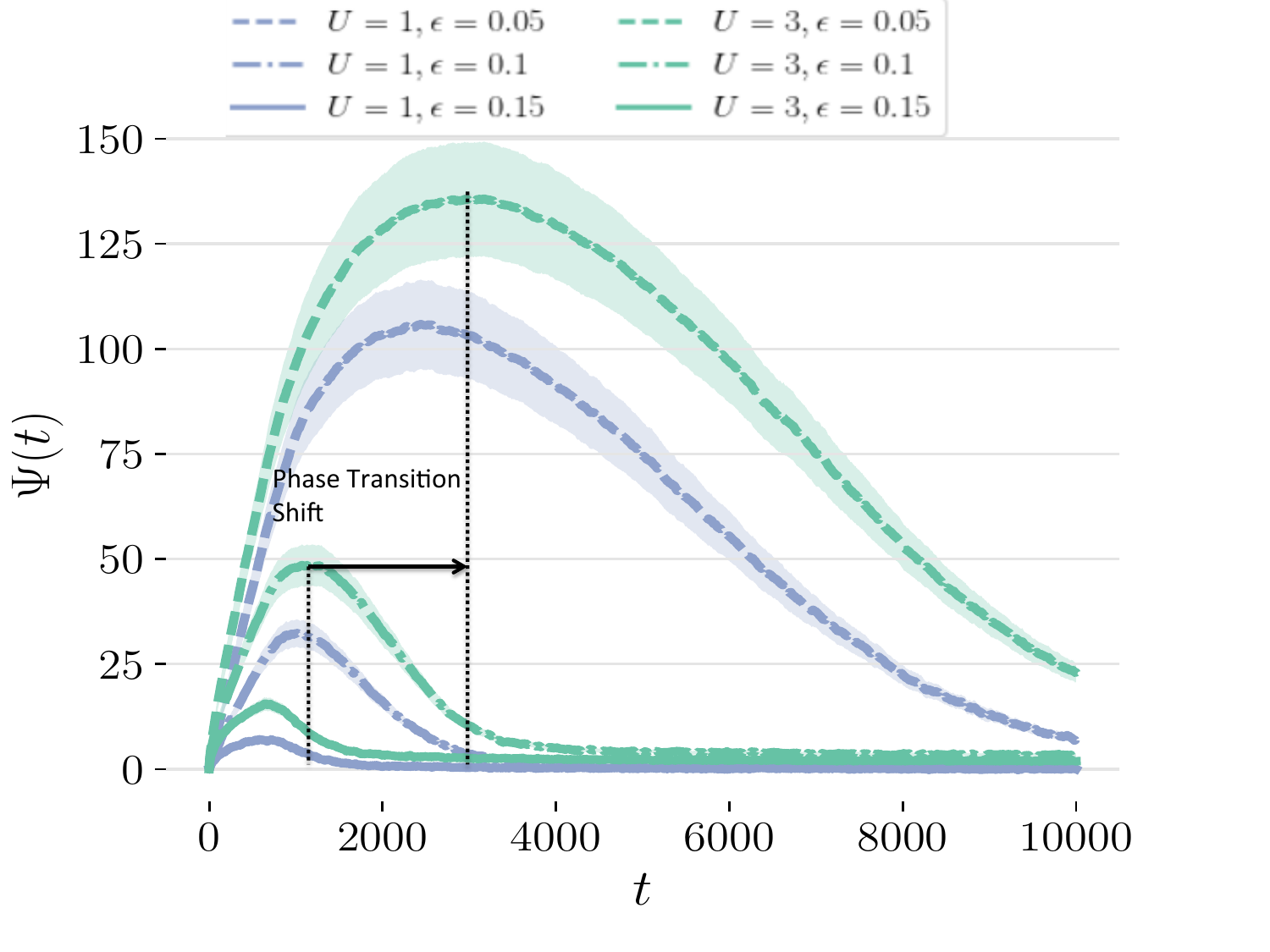}}
\caption{Variation of Queue-regret $\Psi(t)$ with $K, U$ and $\epsilon$ under Q-Ths. The phase-transition point shifts towards the right as $\epsilon$ decreases. The efficiency of learning decreases with increase in the size of the system. We plot the median statistics of the average queue regret over $1000$ simulations and the shaded regions indicate the area between the first and third quartiles.  
\label{fig:kvar} }
\end{center}
\end{figure*}
 Figure~\ref{fig:kvar1} shows the evolution of queue-regret in single queue systems with $5$ and $7$ servers for different values of load $\epsilon$. It can be observed that the regret decays faster in the smaller system as expected from the theoretical upper and lower bounds, which are proportional to the number of servers $K$. In Figure~\ref{fig:kvar2}, we show the evolution of queue regret in systems with $5$ servers and $1,3$ queues for different load values. For the system with $3$ queues, we take the same load for all the queues ($\epsilon_u = \epsilon$ for $u = 1, 2, 3$) and plot the quantity $\Psi(t) = \max_{u} \Psi_u(t)$. We observe that the queue regret decays at a slower rate in the system with $3$ queues as predicted by the theoretical analysis in Theorems~\ref{thm:lb-late-multiq} and~\ref{thm:lb-early-multiq}. 
 
 It is also evident from Figure~\ref{fig:kvar} that the regret of the queueing system grows with decreasing $\epsilon$. We can observe that the time at which the phase transition occurs shifts towards the right with decreasing $\epsilon$ which is predicted by Corollaries~\ref{cor:thm:ucb-best-match-multiq} and \ref{cor:cor:thm:ucb-best-match-multiq}. 


\subsubsection*{Comparison with Implicit Exploration Algorithms.}

\begin{figure}
\begin{center}
{\includegraphics[width=0.65\textwidth]{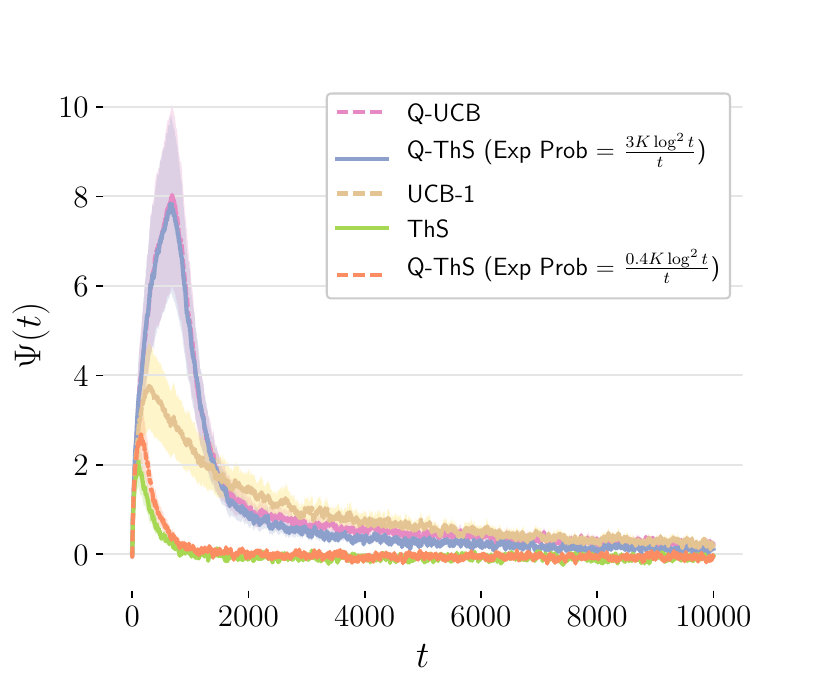}}
\caption{Comparison of queue-regret performance of Q-ThS, Q-UCB, UCB1, and Thompson sampling in a $5$ server system with $\epsilon _u = 0.15$ and $\Delta = 0.17$.  Two variants of Q-ThS are presented, with different exploration probabilities; note that $3K \log^2 t/t$ is the exploration probability suggested by theoretical analysis (which is necessarily conservative).  Tuning the constant significantly improves performance of Q-ThS bringing it closer to Thompson sampling. We plot the median statistics of the average queue regret over $3000$ simulations and the shaded regions in the same color indicate the area between the first and third quartiles.\label{fig:compare}
}
\end{center}
\end{figure}

  In Figure~\ref{fig:compare}, we benchmark the performance of our algorithms against traditional bandit algorithms UCB1 and Thompson sampling. It can be observed that, in the early stage, the traditional algorithms perform better than the proposed ones. This can be explained by the additional forced exploration required by Q-UCB and Q-ThS. In the late stage, we observe that Q-UCB gives slightly better performance than UCB1, however Thompson sampling has better performance in both the stages.  In this numerical example, the results suggest that Thompson sampling is already achieving enough exploration to deliver adequate estimation of service probabilities to stabilize the queues.  We introduced forced exploration as a vehicle to theoretically understand performance limits of queue regret, which brings some loss of performance in the early stage.  Our simulation results thus pose an interesting open direction: to try to quantify the queue regret performance of classical bandit algorithms that do not explicitly include forced exploration.

%% file: conclusion.tex
\section{Discussion and Conclusion}
\label{sec:concl}



This paper provides the first regret analysis of the queueing bandit problem, including a characterization of regret in both early and late stages; and algorithms (Q-UCB, Q-ThS) that are asymptotically optimal (to within poly-logarithmic factors).  Here we also highlight several additional substantial open directions for future work.  

First, note that in the late stage we have a lower bound that is essentially of order $1/t$ (cf.~Proposition \ref{thm:lb-late}), while the upper bound for our algorithm is of order $\log^3 t / t$ (cf.~Proposition \ref{cor:thm:ucb-best-match} and Theorem \ref{thm:ucb-best-match-multiq}).  It remains an open question whether one can close the gap between lower and upper bounds; in particular, this would require development of an algorithm that is able to achieve concentration of the distribution of busy period length with less exploration than the algorithms we have developed.

Second, and related to the previous point, is there a single algorithm that gives optimal performance in {\em both} early and late stages, as well as the optimal switching time between early and late stages?   The price paid for forced exploration by Q-UCB/Q-ThS is an inflation of regret in the early stage.  An important open question is to find a single, adaptive algorithm that gives good performance over all time.  As we note in the Section~\ref{sec:simulations}, classic (no forced exploration) Thompson sampling is an intriguing candidate from this perspective.

Third, the most significant technical hurdle in finding a single optimal algorithm is the difficulty of establishing concentration results for the number of suboptimal arm pulls within a regenerative cycle whose length is dependent on the bandit strategy.  Such concentration results would be needed in two different limits: first, as the start time of the regenerative cycle approaches infinity (for the asymptotic analysis of late stage regret); and second, as the load of the system increases (for the analysis of early stage regret in the heavily loaded regime).  Any progress on the open directions described above would likely require substantial progress on these technical questions as well.

Finally, we conclude by noting that analysis of this problem in an adversarial/non-stochastic setting presents substantial challenges distinct from our analysis.  At its core, the challenge is that even the formulation of an appropriate adversarial/non-stochastic setting is nontrivial.  Suppose in particular that service rates are adversarial/non-stochastic, but the arrivals are a stationary stochastic process. Immediately within such a context, we are confronted with difficulty in defining both an appropriate notion of stability, as well as an appropriate notion of a genie policy.  In the standard bandit problem in this setting, the goal of the scheduler is to match a genie which chooses the arm with the maximum cumulative (long term) reward; but this may not be appropriate in the queue setting.  More precisely, two different regimes become apparent.  In particular, if the long run average service rate of every single server is insufficient to serve the long-run average arrival rate, then the queue blows up, and optimal queue-regret (against the best arm in hindsight) is identical to the usual multi-armed bandit as $O(\sqrt{t})$.  On the other hand, if all the servers give sufficient service to empty the queue regularly, then decisions made in the ``recent'' past will dominate queue-regret, and therefore we will want to compare against a genie that adaptively chooses the best server locally in time depending on the service capacity realized during the current busy period.

%% file: proof-ub.tex
\section{Proof of Regret Upper Bound}
\label{sec:proofs-ub}

\subsection{Proof of Lemma~\ref{lem:easy-ub}}
\label{subsec:proof-easy-ub}
\subsubsection*{Alternate Coupled Service Process.}
 We first construct an alternate service process such that, under any scheduling policy, the queue evolution for this system has the same distribution as that for the original system.  We will use this construction of the service process to prove {\em (i)} the upper bound in Lemma~\ref{lem:easy-ub} and {\em (ii)} the lower bound in Theorem~\ref{thm:lb-late-multiq}. The construction allows coupling of service offered to each queue across different servers. Specifically, the service process satisfies the following condition: {\em for each queue, the service offered by different servers at any time-slot could possibly be dependent on each other but has the same marginal distribution as that in the original system and is independent of the service offered to other queues.}
 
More precisely, the service process is constructed as follows: for each $u \in [U]$, let $\left\lbrace U_u(t) \right\rbrace_{t \geq 1}$ be i.i.d.\ random variables distributed uniformly in $(0,1)$. The service process for queue $u$ and server $k$ be given by $R_{uk}(t) = \mathds{1}\left\lbrace U_u(t) \leq \mu_{uk} \right\rbrace \; \forall t.$ Since $\EE\left[ R_{uk}(t) \right] = \mu_{uk}$, the marginals of the service offered by each of the servers is the same as the original system, and also, service offered to different queues are independent of each other. 

Note that the evolution of the queues is a function of the process $\left( \mathbf{Z}(l) \right)_{l \geq 1} := \left( \mathbf{A}(l), \pmb{\kappa}(l), \mathbf{S}(l) \right)_{l \geq 1}.$ To prove that the process $\mathbf{Z}$ has the same distribution  for both service processes -- independent Bernoulli and the coupled process described above -- we use induction on the size of the finite-dimensional distribution of $\mathbf{Z}$. In other words, we show that the distribution of the vector $\left( \mathbf{Z}(l) \right)_{l=1}^{t}$ is the same for the two systems for all $t$ by induction on $t$. 
 
 Suppose that the hypothesis is true for $t-1$. Now consider the conditional distribution of $\mathbf{Z}(t)$ given $\left( \mathbf{Z}(l) \right)_{l=1}^{t-1}.$ Given $( \mathbf{Z}(l) )_{l=1}^{t-1}$, the distribution of $\left( \mathbf{A}(t), \pmb{\kappa}(t) \right)$ is identical for the two systems for any scheduling policy since the two systems have the same arrival process. Also, given $\left( (\mathbf{Z}(l) )_{l=1}^{t-1}, \mathbf{A}(t), \pmb{\kappa}(t)\right)$, the distribution of $\mathbf{S}(t)$ depends only on the marginal distribution of the scheduled servers given by $\pmb{\kappa}(t)$ which is again the same for the two systems. Therefore, $\left( \mathbf{Z}(l) \right)_{l=1}^{t}$ has the same distribution in the two systems. Since the statement is true for $t=1$, it is true for all $t$.

Consequently, w.r.t. the process $\mathbf{Z}$, the original system with independent Bernoulli process is stochastically equivalent to the alternate system with the coupled service process. Notably, the queue-regret in the two systems are equal, and so is the traditional regret (given by the cumulative rate loss). 
\begin{proof}[Proof of Lemma~\ref{lem:easy-ub}.]
We will prove the upper bound using the coupled service process described above. Consider any queue $u \in [U]$.  Since $\mu^*_u > \mu_{uk} \; \forall k \neq k^*_u$, we have $R_{uk^*_u}(t) \geq R_{uk}(t) \; \forall k \neq k^*_u, \, \forall t$. This implies that $S^*_u(t) \geq S_u(t)$ and $Q^*_u(t) \leq Q_u(t) \; \forall t$. Using this and the fact that the difference in queue-lengths is equal to the difference in the total number of departures, we have
\begin{align*}
Q_u(t) - Q^*_u(t) & = \sum_{l=1}^t S^*_u(l)\ind{Q^*_u(l-1) > 0} - S_u(l)\ind{Q_u(l-1) > 0}	\\
& \leq \sum_{l=1}^t \left( S^*_u(l) - S_u(l) \right)\ind{Q_u(l-1) > 0}	\\
& \leq \sum_{l=1}^t S^*_u(l) - S_u(l).
\end{align*}
This gives us the required result, i.e., $$\EE\left[ \bm{Q}(t) - \bm{Q}^*(t) \right] \leq \sum_{l=1}^{t} \EE\left[ \bm{S}^*(l) - \bm{S}(l) \right].$$ 
\end{proof}

\subsection{Proof of Theorem~\ref{thm:ucb-best-match-multiq}}
As mentioned, we prove all our results for the generalized setting of a switch network with a unique optimal matching. We state and prove a few intermediate lemmas that are useful in proving Theorem~\ref{thm:ucb-best-match-multiq}. Proof of Lemma~\ref{lem:exploit-ub} is given in separate parts for Q-UCB and Q-ThS. All other proofs in Appendix~\ref{sec:proofs-ub}, although given for the algorithm Q-UCB, follow in an identical manner for Q-ThS. Important notation used in the two algorithms are summarized in Table~\ref{tab:notation-alg}.

As shown in Algorithm~\ref{alg:ucb}, $\mathsf{E}(t)$ indicates whether Q-UCB chooses to explore at time $t$. We now obtain a bound on the expected number of time-slots Q-UCB chooses to explore in an arbitrary time interval $(t_1, t_2]$. Since at any time $t$, Q-UCB decides to explore with probability $\min\{1,3K\frac{\log^2t}{t}\},$ we have
\begin{align}
\label{eq:explore-mean-ub}
\EE\left[ \sum_{l = t_1 + 1}^{t_2} \mathsf{E}(l) \right] 
\leq 3K \sum_{l = t_1 + 1}^{t_2} \frac{\log^2l}{l} \leq 3K  \int_{t_1}^{t_2} \frac{\log^2l}{l} \diff l = K\left(\log^3 t_2 - \log^3 t_1 \right).
\end{align}
The following lemma gives a probabilistic upper bound on the same quantity.
Recall that $w(t) = t^{(1 - 1/\beta)}$ as defined in \cref{thm:ucb-best-match-multiq}. We will use this definition of $w(t)$ for all the lemmas in this section.
\begin{lemma}
\label{lem:explore-ub}
\begin{enumerate}[label=(\alph*)]
\item \label{lem:explore-ub1}
For any $t$ and $t_1 < t_2$, 
\begin{align*}
\PP\left[  \sum_{l=t_1+1}^{t_2}  \mathsf{E}(l) \geq 5 \max \left( \log t, K\left( \log^3 t_2 - \log^3 t_1 \right) \right) \right]
& \leq \frac{1}{t^4}.
\end{align*}
\item \label{lem:explore-ub2}
For $t \geq 5.8 \times 10^3$, $$\PP\left[ \sum_{l=1}^{t} \mathsf{E}(l) > Kw(t) \right] \leq \frac{1}{t^{2K}}.$$
\end{enumerate}
\end{lemma}
To prove the result, we will use the following Chernoff bound: for a sum of independent Bernoulli random variables $Y$ with mean $\EE Y$ and for any $\delta > 0,$
\begin{align}
\label{eq:chernoff}
\PP\left[ Y \geq (1+\delta) \EE Y \right] \leq \left(  \frac{e^\delta}{(1+\delta)^{1+\delta}} \right)^{\EE Y}.
\end{align}
\begin{proof}[Proof of Lemma~\ref{lem:explore-ub}\ref{lem:explore-ub1}.]
If $\EE Y \geq \log t,$ the above bound for $\delta = 4$ gives 
$$\PP\left[ Y \geq 5 \EE Y \right] \leq \frac{1}{t^4}.$$ 
Note that $\left\lbrace \mathsf{E}(l)\right\rbrace_{l=t_1 + 1}^{t_2}$ are independent Bernoulli random variables and let $X = \sum_{l=t_1}^{t_2}\mathsf{E}(l).$  Now consider the probability $\PP\left[ X \geq 5\max \left( \log t, \EE X \right) \right].$ If $\EE X \geq \log t,$ then the result is true from the above Chernoff bound. If $\EE X < \log t,$ then it is possible to construct a random variable $Y$ which is a sum of independent Bernoulli random variables, has mean $\log t$ and stochastically dominates $X,$ in which case we can again use the Chernoff bound on $Y$. Therefore, $$\PP\left[ X \geq 5 \log t \right] \leq \PP\left[ Y \geq 5 \log t \right] \leq \frac{1}{t^4}.$$

Using inequality~\eqref{eq:explore-mean-ub}, we have the required result, i.e., 
\begin{align*}
\PP\left[  \sum_{l=t_1+1}^{t_2}  \mathsf{E}(l) \geq 5 \max \left( \log t, K\left( \log^3 t_2 - \log^3 t_1 \right) \right) \right]	 \leq \PP\left[ X \geq 5\max \left( \log t, \EE X \right) \right] \leq 1/t^4.
\end{align*}
\end{proof}

\begin{proof}[Proof of Lemma~\ref{lem:explore-ub}\ref{lem:explore-ub2}.]
Using the following alternate form of the Chernoff bound \eqref{eq:chernoff}
\begin{align*}
\PP\left[ Y > y \right] \leq \frac{1}{e^{\EE Y}} \left( \frac{e \EE Y}{y} \right)^y,
\end{align*}
we get,
\begin{align*}
\PP\left[ \sum_{l=1}^{t} \mathsf{E}(l) > Kw(t) \right] & \leq \PP\left[ \sum_{l=1}^{t} \mathsf{E}(l) > K \exp\left( \left( 2 \log t \right)^{2/3} \right) \right]	\\
& \leq \frac{1}{e^{K \log^3 t}} \left( \frac{e \log^3 t}{\exp\left( \left( 2 \log t \right)^{2/3} \right)} \right)^{K \exp\left( \left( 2 \log t \right)^{2/3} \right)}	\\
& \leq \frac{1}{t^{2K}}
\end{align*}
for all $t \geq 5.8 \times 10^3$. To justify the last inequality, we first verify that the function $$e^{(2x)^{2/3}}\left( (2x)^{2/3} - 1 - 3\log x \right) + x^3 - 2x$$ is positive for all $x \geq \log(5800)$. The result then follows by taking $x = \log t.$
\end{proof}
The next lemma shows that, with high probability, Q-UCB (or Q-ThS) does not schedule a sub-optimal matching when it exploits in the late stage.
\begin{lemma}
\label{lem:exploit-ub}
Define the event \event{ev:no-suboptimal}$\mathcal{E}_{\ref{ev:no-suboptimal}} := \sum_{l=w(t)+1}^{t} \sum_{u \in [U]} \sum_{k\neq k^*_u} \mathsf{I}_{uk}(l) = 0.$ Then, for Q-UCB (Q-ThS),
$$\PP\left[ \mathcal{E}_{\ref{ev:no-suboptimal}}^c \right] = \PP\left[ \bigcup_{u \in [U]} \sum_{l=w(t)+1}^{t} \sum_{k\neq k^*_u} \mathsf{I}_{uk}(l) > 0 \right] \leq \frac{UK}{6t^3},$$
for all $t \geq \tau_1$ ($t \geq \tau_2$), $t \geq  \exp \left( \frac{4}{\Delta^2 (1 - 1/\beta)^3} \right)$.
\end{lemma}
\begin{proof}[Proof.]
Let $X_{uk}(l), u = 1, 2, \dots, K \, , k = 1, 2, \dots, K, \, l = 1, 2, 3, \dots$ be independent random variables denoting the service offered in the $l^{th}$ assignment of server $k$ to queue $u$ and let $B_{uk}(s,t) = \frac{1}{s} \sum_{l=1}^s X_{uk}(l) + \sqrt{\frac{\log^2 t}{2 s}}.$ Consider the events
\begin{align}
\label{ev:explore-lb}
T_{uk}(l) \geq \frac{1}{2} \log^3(l-1) \quad \forall  k \in [K],  u \in [U], w(t)+1 \leq l \leq t,
\end{align}
\begin{align}
\label{ev:opt-high1}
B_{uk^*_u}(s,l) > \mu^*_u \quad \forall s, l \text{ s.t. } \frac{1}{2} \log^3(w(t)) \leq s \leq t-1, w(t)+1 \leq l \leq t,  \; \forall u \in [U],
\end{align}
and
\begin{align}
\label{ev:subopt-low1}
B_{uk}(s,l) \leq \mu^*_u \quad \forall s, l \text{ s.t. } \frac{1}{2} \log^3(w(t)) \leq s \leq t-1, w(t)+1 \leq l \leq t,  \forall k \neq k^*_u, \; \forall u \in [U].
\end{align}
It can be seen that, given the above events, Q-UCB schedules the optimal matching in all time-slots in $(w(t), t]$ in which it decides to exploit, i.e., $\sum_{l=w(t)+1}^{t} \sum_{k\neq k^*_u} \mathsf{I}_{uk}(l) = 0$ for all $u \in [U]$. We now show that the events above occur with high probability. 

Note that, since the matchings in $\mathcal{X}$ cover all the links in the system, $T_{uk}(l+1) \leq \frac{1}{2} \log^3(l)$ for some $u, k$ implies that  $\sum_{s=1}^{l} \mathds{1} \left\lbrace \pmb{\kappa}(s) = \pmb{\kappa} \right\rbrace  \leq \frac{1}{2} \log^3(l)$ for some $\pmb{\kappa} \in \mathcal{E}.$ Since $\sum_{s=1}^{l} \mathds{1} \left\lbrace \pmb{\kappa}(s) = \pmb{\kappa} \right\rbrace$ is a sum of i.i.d.\ Bernoulli random variables with sum mean at least $\log^3 (l),$ we use Chernoff bound to prove that event \eqref{ev:explore-lb} occurs with high probability. Note that $t \geq  \exp \left( \frac{4}{\Delta^2 (1 - 1/\beta)^3} \right)$ implies that $\log^3(w(t)) = (1 - 1/\beta)^3 \log^3t \geq \left( \frac{2 \log t}{\Delta} \right)^{2}$. Therefore,
\begin{equation}
\label{eq:lb-w(t)}
\log(w(t)) \geq \left( \frac{2 \log t}{\Delta} \right)^{2/3}.
\end{equation} 
This gives us
\begin{align}
\label{eq:explore-lb}
\PP\left[ \text{\eqref{ev:explore-lb} is false}	\right]
& \leq \sum_{\pmb{\kappa} \in \mathcal{E}} \sum_{l = w(t)}^{t-1} \PP\left[ \sum_{s=1}^{l}  \mathds{1} \left\lbrace \pmb{\kappa}(s) = \pmb{\kappa} \right\rbrace \leq \frac{1}{2} \log^3(l) \right]	\nonumber	\\
& \leq K t \exp\left( -\frac{1}{8} \log^3(w(t)) \right)	\nonumber	\\
& \leq K t \exp\left( -\frac{1}{8} \left( \frac{2 \log t}{\Delta} \right)^{2} \right)	\nonumber	\\
& \leq K t \exp\left( -\frac{1}{2} \log^2 t \right).
\end{align}
Similarly, probability of events \eqref{ev:opt-high1} and \eqref{ev:subopt-low1} can be bounded as follows --
\begin{align}
\label{eq:conc-opt-high1}
\PP\left[ \text{\eqref{ev:opt-high1} is false}	\right]
& \leq \sum_{ u \in [U]} \sum_{l=w(t) + 1}^{t}\sum_{s=\frac{1}{2} \log^3(w(t))}^{t-1} \PP\left[ B_{uk^*_u}(s,l) \leq \mu^*_u \right]	\nonumber	\\
& \leq U \sum_{l=w(t) + 1}^{t}\sum_{s=\frac{1}{2} \log^3(w(t))}^{t-1} \exp \left( -\log^2(l) \right)	\nonumber	\\
& \leq U t^2 \exp \left( - \log^2(w(t)) \right)	\nonumber	\\
& \leq U t^2 \exp \left( - \left( \frac{2 \log t}{\Delta} \right)^{4/3} \right)	\nonumber	\\
& \leq U t^2 \exp \left( - \left( 2 \log t \right)^{4/3} \right).
\end{align}

\begin{align}
\label{eq:conc-subopt-low1}
\PP\left[ \text{\eqref{ev:subopt-low1} is false}	\right]
& \leq \sum_{u \in [U], k\neq k^*_u} \sum_{l=w(t) + 1}^{t}\sum_{s=\frac{1}{2} \log^3(w(t))}^{t-1} \PP\left[ B_{uk}(s,l) > \mu^*_u \right]	\nonumber 	\\
& \leq \sum_{u \in [U], k\neq k^*_u}  \sum_{l=w(t) + 1}^{t} \sum_{s=\frac{1}{2} \log^3(w(t))}^{t-1}  \PP\left[ B_{uk}(s,l) > \Delta + \mu_{uk} \right]	\nonumber 	\\ 
& \leq UK \sum_{l=w(t) + 1}^{t} \sum_{s=\frac{1}{2} \log^3(w(t))}^{t-1}  \exp \left( -2s \left( \Delta - \sqrt{\frac{\log^2 l}{2s}} \right)^2	 \right)	\nonumber 	\\
& \leq UK  t^2 \exp \left( - \log^3 (w(t)) \left( \Delta - \sqrt{\frac{\log^2 t}{\log^3 (w(t))}} \right)^2	 \right)	\nonumber 	\\
& \leq UK t^2 \exp \left( -  \log^2 t \right).
\end{align}
Combining the inequalities~\eqref{eq:explore-lb}, \eqref{eq:conc-opt-high1} and \eqref{eq:conc-subopt-low1} we have
\begin{align*}
\PP\left[ \mathcal{E}_{\ref{ev:no-suboptimal}}^c \right]
& \leq  \PP\left[ \text{\eqref{ev:explore-lb} is false}	\right]	+ \PP\left[ \text{\eqref{ev:opt-high1} is false}	\right] + \PP\left[ \text{\eqref{ev:subopt-low1} is false}	\right]	\\
& \leq K t \exp\left( -\frac{1}{2} \log^2 t \right) + U t^2 \exp \left( - \left( 2 \log t \right)^{4/3} \right) + UK t^2 \exp \left( -  \log^2 t \right)	\numberthis \label{eq:subopt-prob-ucb}	\\
& \leq U K t \left( \exp\left( -\frac{1}{2} \log^2 t \right) +  \frac{t}{2} \exp \left( - \left( 2 \log t \right)^{4/3} \right) + t \exp \left( -\log^2 t \right) \right)	\\
& \leq \frac{UK}{6t^3} 
\end{align*}
$\forall \, t \geq 5800.$
This proves the result for Q-UCB.

The proof of this result for Q-ThS follows in a similar fashion. For Q-ThS, events \eqref{ev:opt-high1} and \eqref{ev:subopt-low1} are substituted with the following events
\begin{equation}
\label{event:e2}
\theta _{uk^*_u} (s) > \mu_u ^* - \frac{ \log  (s-1)}{\sqrt{2T_{uk^*_u}(s)}}, \mbox{  } \forall s, \mbox{ s.t. } w(t)+1 \leq s \leq t , u \in [U]  
\end{equation}
and
\begin{equation}
\label{event:e3}
\theta _{uk} (s) \leq  \mu_u ^* - \frac{ \log  (s-1)}{\sqrt{2T_{uk^*_u}(s)}}, \mbox{  } \forall s, k \mbox{ s.t. } w(t)+1 \leq s \leq t , k \neq k^*_u, u \in [U].
\end{equation}
It is then sufficient to prove that the above two events occur with high probability. Given events \eqref{ev:explore-lb}, \eqref{event:e2}, \eqref{event:e3}, Q-ThS schedules the optimal matching in all time-slots in $(w(t), t]$ in which it decides to exploit, i.e., $\sum_{l=w(t)+1}^{t} \sum_{k\neq k^*_u} \mathsf{I}_{uk}(l) = 0$ for all $u \in [U]$.
 
Let $\Sigma_{u,k,l} = \sum_{r = 1}^{l} X_{uk}(r)$ and $S_{uk}(l) = \hat{\mu}_{uk}(l)T_{uk}(l) = \Sigma_{u,k,T_{uk}(l)}$ for all $u\in [U], k \in [K]$, $l \in \setN$. We use the `Beta-Binomial trick' (used in \cite{kaufmann2012thompson,agrawal-goyal12thompson}), which gives a relation between the c.d.fs of Beta and Binomial distributions to prove the high probability results.  Let $F^{\mathrm{Beta}}_{a,b}$ and $F^{\mathrm{B}}_{n,p}$ denote the c.d.f of $\mathrm{Beta}(a,b)$ distribution and the c.d.f. of $\mathrm{Binomial}(n,p)$ distribution respectively. Then
\begin{align*}
F^{\mathrm{Beta}}_{a,b}(y) = 1 - F^{\mathrm{B}}_{a+b-1, y}(a-1).
\end{align*} 
For each $s, l \in \setN$, let $\{Z_{s,l}(r)\}_{r>0}$ be a sequence of i.i.d.\ Bernoulli random variables with mean $\mu_u ^* - \frac{ \log  (s-1)}{\sqrt{2l}}$. Now, to bound probability of event~\eqref{event:e2}, 
\begin{align*}
\PP\left[ \text{\eqref{event:e2} is false} \bigcap \eqref{ev:explore-lb} \text{ is true} \right]
& \leq \sum_{u \in [U]} \sum_{s = w(t)+1}^{t} \PP \left[ \theta _{uk^*_u} (s) \leq \mu_u ^* - \frac{ \log  (s-1)}{\sqrt{2T_{uk^*_u}(s)}} \bigcap \eqref{ev:explore-lb} \text{ is true} \right].
\end{align*}
Now, for any $u \in [U]$, $w(t)+1 \leq s \leq t$, we have
\begin{align*}
& \PP \left[ \theta _{uk^*_u} (s) \leq \mu_u ^* - \frac{ \log  (s-1)}{\sqrt{2T_{uk^*_u}(s)}} \bigcap \eqref{ev:explore-lb} \text{ is true} \right]	\\
& = \sum_{l = \frac{1}{2}\log^3(s-1)}^{s} \EE\left[ \ind{T_{uk^*_u}(s) = l} F^{\mathrm{Beta}}_{S_{uk^*_u}(s)+1,T_{uk^*_u}(s)-S_{uk^*_u}(s)+1} \left( \mu_u ^* - \frac{ \log  (s-1)}{\sqrt{2T_{uk^*_u}(s)}} \right) \right] \\
& = \sum_{l = \frac{1}{2}\log^3(s-1)}^{s}  \EE\left[ 1 - F^{\mathrm{B}}_{l+1,\mu_u ^* - \frac{ \log  (s-1)}{\sqrt{2l}}} \left( \Sigma_{u,k^*_u,l} \right)  \right]  \\
&  \leq \sum_{l = \frac{1}{2}\log^3(s-1)}^{s}  \PP \left[ \Sigma_{u,k^*_u,l} \leq \sum_{r=1}^{l+1} Z_{s,l}(r) \right]  \\
& \leq \sum_{l = \frac{1}{2}\log^3(s-1)}^{s} \exp\left( -\frac{\log^2(s-1)}{4} \right)	\\
& \leq t \exp \left( -\frac{1}{4} \log^2 (w(t)) \right)	\\
& \leq t \exp \left( -\frac{1}{4} \left( 2 \log t \right)^{4/3} \right).
\end{align*}
The last inequality follows by using \eqref{eq:explore-lb} to bound the first term and Chernoff-Hoeffding inequality to bound the second term.

Similarly, the probability of event~\eqref{event:e3} can be bounded as follows.
\begin{align*}
\PP \left[ \text{\eqref{event:e3} is false} \bigcap \eqref{ev:explore-lb} \text{ is true} \right] 
& \leq \sum_{u \in [U], k \neq k^*_u} \sum_{s = w(t)+1}^{t} \PP \left[ \theta _{uk} (s) > \mu_u ^* - \frac{ \log  (s-1)}{\sqrt{2T_{uk^*_u}(s)}} \bigcap \eqref{ev:explore-lb} \text{ is true} \right]  \\
\end{align*}
Now, for any $u \in [U], k \neq k^*_u$, $w(t)+1 \leq s \leq t$, we have
\begin{align*}
& \PP \left[ \theta _{uk} (s) > \mu_u ^* - \frac{ \log  (s-1)}{\sqrt{2T_{uk^*_u}(s)}} \bigcap \eqref{ev:explore-lb} \text{ is true}  \right]	\\
& \leq \sum_{l = \frac{1}{2}\log^3(s-1)}^{s} \PP\left[ T_{uk^*_u}(s) = l \right] \EE\left[ 1 - F^{\mathrm{Beta}}_{\Sigma_{u,k,l}+1,l-\Sigma_{u,k,l}+1} \left( \mu_u ^* - \frac{ \log  (s-1)}{\sqrt{2l}} \right) \given[\bigg] T_{uk^*_u}(s) = l  \right]	\\
& = \sum_{l = \frac{1}{2}\log^3(s-1)}^{s} \PP\left[ T_{uk^*_u}(s) = l \right] \EE\left[ F^{\mathrm{B}}_{l+1,\mu_u ^* - \frac{ \log  (s-1)}{\sqrt{2l}}} \left( \Sigma_{u,k,l} \right) \given[\bigg] T_{uk^*_u}(s) = l  \right]	\\
& \leq \sum_{l = \frac{1}{2}\log^3(s-1)}^{s} \PP \left[ \sum_{r=1}^{l+1} Z_{s,l}(r) \leq \Sigma_{u,k,l} \right]  \\
& \leq \sum_{l = \frac{1}{2}\log^3(s-1)}^{s}  \exp \left( -\frac{l}{2} \left( \Delta - \sqrt{\frac{\log^2 (s-1)}{2l}} \right)^2	 \right) 	\\
& \leq  t \exp \left( -\frac{1}{4} \log^3 (w(t)) \left( \Delta - \sqrt{\frac{\log^2 t}{\log^3 (w(t))}} \right)^2	 \right)	\\
& \leq t \exp \left( - \frac{1}{4} \log^2 t \right).	
\end{align*}

Combining the above inequalities and \eqref{eq:explore-lb}, we have
\begin{align*}
\PP\left[ \mathcal{E}_{\ref{ev:no-suboptimal}}^c \right]
& \leq  \PP\left[ \text{\eqref{ev:explore-lb} is false}	\right]	+ \PP\left[ \text{\eqref{event:e2} is false} \bigcap \eqref{ev:explore-lb} \text{ is true} \right] + \PP \left[ \text{\eqref{event:e3} is false} \bigcap \eqref{ev:explore-lb} \text{ is true} \right] 	\\
& \leq K t \exp\left( -\frac{1}{2} \log^2 t \right) + U t^2 \exp \left( - \frac{1}{4} \left( 2 \log t \right)^{4/3} \right) + UK t^2 \exp \left( - \frac{1}{4} \log^2 t \right).	\numberthis \label{eq:subopt-prob-ths}	 
\end{align*}
Note from \eqref{eq:large-const-tau} that the last expression is $\leq \frac{UK}{6t^3}$ for $t \geq \tau_2$, which proves the result for Q-ThS. 
\end{proof}

For any time $t$, let $$B_u(t) := \min \{s \geq 0 : Q_u(t-s) = 0\}$$ denote the time elapsed since the beginning of the current regenerative cycle for queue $u$. Alternately, at any time $t$, $t-B_u(t)$ is the last time instant at which queue $u$ was zero. 

The following lemma gives an upper bound on the sample-path queue-regret in terms of the number of sub-optimal schedules in the current regenerative cycle.
\begin{lemma}
\label{lem:queue-diff-characterization}
For any $t \geq 1$,
$$Q_u(t) - Q^*_u(t) \leq  \sum_{l=t-B_u(t)+1}^{t} \left( \mathsf{E}(l) + \sum_{k\neq k^*_u} \mathsf{I}_{uk}(l) \right).$$
\end{lemma}
\begin{proof}[Proof.]
If $B_u(t) = 0$, i.e., if $Q_u(t) = 0,$ then the result is trivially true. 

Consider the case where $B_u(t) > 0$. Since $Q_u(l) > 0$ for all $t-B_u(t)+1 \leq l \leq t$, we have 
\begin{align*}
Q_u(l) = Q_u(l-1) + A_u(l) - S_u(l)	\quad \forall t-B_u(t)+1 \leq l \leq t.
\end{align*}
This implies that $$Q_u(t) = \sum_{l=t-B_u(t)+1}^{t} A_u(l) - S_u(l).$$ 
Moreover,
\begin{align*}
Q_u^*(t) = \max_{1 \leq s \leq t} \left( Q_u^*(0) + \sum_{l=s}^{t} A_u(l) - S^*_u(l) \right)^+ \geq \sum_{l=t-B_u(t)+1}^{t} A_u(l) - S^*_u(l).
\end{align*}
Combining the above two expressions, we have
\begin{align*}
Q_u(t) - Q^*_u(t) & \leq \sum_{l=t-B_u(t)+1}^{t} S^*_u(l) - S_u(l)	\\
& =  \sum_{l=t-B_u(t)+1}^{t} \sum_{k \in [K]} \left( R_{u k^*_u}(l) - R_{uk}(l) \right) \left( \mathsf{E}_{uk}(l) + \mathsf{I}_{uk}(l) \right)	\\
& \leq \sum_{l=t-B_u(t)+1}^{t} \sum_{k\neq k^*_u} \left( \mathsf{E}_{uk}(l) + \mathsf{I}_{uk}(l) \right)	\\
& \leq \sum_{l=t-B_u(t)+1}^{t} \left( \mathsf{E}(l) + \sum_{k\neq k^*_u} \mathsf{I}_{uk}(l) \right),
\end{align*}
where the second inequality follows from the assumption that the service provided by each of the links is bounded by $1,$ and the last inequality from the fact that $\sum_{k \in [K]} \mathsf{E}_{uk}(l) = \mathsf{E}(l) \; \forall l, \forall u \in [U].$ 
\end{proof}

In the next lemma, we derive a coarse high probability upper bound on the queue-length. This bound on the queue-length is used later to obtain a first cut bound on the length of the regenerative cycle in Lemma~\ref{lem:busy-period-ub1}.
\begin{lemma}
\label{lem:ub-queue1}
Define the event \event{ev:explore-ub1}$\mathcal{E}_{\ref{ev:explore-ub1}} := \left\lbrace \sum_{l=1}^{t} \mathsf{E}(l) \leq Kw(t) \right\rbrace$. Then for any $l \in [1,t]$,
$$\PP\left[ \left\lbrace Q_u(l) > 2 K w(t) \right\rbrace \cap \mathcal{E}_{\ref{ev:no-suboptimal}} \cap \mathcal{E}_{\ref{ev:explore-ub1}} \right] \leq \frac{1}{t^3}$$
$\forall t$ s.t.  $\frac{ w(t)}{\log t} \geq \frac{2}{\epsilon_u}$.
\end{lemma}
\begin{proof}[Proof.]
We show this result for $l = t$ but the same argument holds for any  $l \in [1,t]$. From Lemma~\ref{lem:queue-diff-characterization},
\begin{align*}
 Q_u(t) - Q^*_u(t)	
\leq  \sum_{l=t-B_u(t)+1}^{t} \left( \mathsf{E}(l) + \sum_{k\neq k^*_u} \mathsf{I}_{uk}(l) \right)
\leq  \sum_{l=1}^{t} \left( \mathsf{E}(l) + \sum_{k\neq k^*_u} \mathsf{I}_{uk}(l) \right).
\end{align*}
Since $Q^*_u(t)$ is distributed according to $\pi_{(\lambda_u, \mu^*_u)},$
\begin{align*}
\PP\left[ Q^*_u(t) > w(t)  \right] = \frac{\lambda_u}{\mu^*_u}\left( \frac{\lambda_u\left( 1-\mu^*_u \right)}{\left( 1-\lambda_u \right)\mu^*_u} \right)^{w(t)} \leq \exp \left( w(t) \log \left( \frac{\lambda_u\left( 1-\mu^*_u \right)}{\left( 1-\lambda_u \right)\mu^*_u} \right) \right) \leq  \frac{1}{t^3}
\end{align*}
if $\frac{ w(t)}{\log t} \geq \frac{2}{\epsilon_u}.$ The last inequality follows from the following bound --
\begin{align*}
\log \left( \frac{\left( 1-\lambda_u \right)\mu^*_u}{\lambda_u\left( 1-\mu^*_u \right)} \right) & = \log \left( 1 + \frac{\epsilon_u}{\lambda_u\left( 1-\mu^*_u \right)} \right)	\\
& \geq  \log \left( 1 + 4 \epsilon_u \right) \quad \text{since } \left( \lambda_u\left( 1-\mu^*_u \right) < 1/4 \right)	\\
& \geq  \frac{3}{2} \epsilon_u.
\end{align*}
Now, note that, given $\mathcal{E}_{\ref{ev:no-suboptimal}}$, $$\sum_{l=1}^{t} \sum_{k\neq k^*_u} \mathsf{I}_{uk}(l) \leq (K-1) w(t) + \sum_{l=w(t)+1}^{t} \sum_{k\neq k^*_u} \mathsf{I}_{uk}(l) = (K-1) w(t).$$ 
Using the inequalities above, we have
\begin{align*}
\PP\left[ \left\lbrace Q_u(l) > 2 K w(t) \right\rbrace \cap \mathcal{E}_{\ref{ev:no-suboptimal}} \cap \mathcal{E}_{\ref{ev:explore-ub1}} \right] 
\leq \PP\left[ Q^*_u(t) >  w(t) \right]  \leq \frac{1}{t^3}.
\end{align*}
\end{proof}

\begin{lemma}
\label{lem:busy-period-ub1}
Let $v'_u(t) = \frac{6 K}{\epsilon_u} w(t)$ and let $v_u$ be an arbitrary function. Then,
 $$\PP\left[ \left\lbrace B_u\left( t-v_u(t) \right) > v'_u(t) \right\rbrace \cap \mathcal{E}_{\ref{ev:no-suboptimal}} \cap \mathcal{E}_{\ref{ev:explore-ub1}} \right] \leq \frac{2}{t^3}$$
 $\forall t$ s.t.  $\frac{ w(t)}{\log t} \geq \frac{2}{\epsilon_u}$ and $v_u(t) + v'_u(t) \leq t/2$.
\end{lemma}
\begin{proof}[Proof.]
Let $r(t) := t-v_u(t).$ Consider the events
\begin{align}
Q_u(r(t)-v'_u(t)) & \leq 2 K w(t),	\label{ev:queue-small-ub1}	\\
\sum_{l=r(t)-v'_u(t)+1}^{r(t)} A_u(l) -  R_{uk^*_u}(l) & \leq -\frac{\epsilon_u}{2} v'_u(t),	\label{ev:arrival-capacity-ub1}	\\
\sum_{l=r(t)-v'_u(t)+1}^{r(t)} \mathsf{E}(l) + \sum_{k\neq k^*_u} \mathsf{I}_{uk}(l) & \leq K w(t).		\label{ev:subopt-schedule-ub1} 
\end{align}
By the definition of $v'_u(t),$
\begin{align*}
2 K w(t)  -\frac{\epsilon_u}{2} v'_u(t) \leq -  K w(t).
\end{align*}

Given Events~\eqref{ev:queue-small-ub1}-\eqref{ev:subopt-schedule-ub1}, the above inequality implies that
\begin{align*}
Q_u(r(t)-v'_u(t)) + \sum_{l=r(t)-v'_u(t)+1}^{r(t)} A_u(l) & \leq \sum_{l=r(t)-v'_u(t)+1}^{r(t)} R_{uk^*_u}(l) -  \left( \mathsf{E}(l) + \sum_{k\neq k^*_u} \mathsf{I}_{uk}(l) \right)	\\
& \leq \sum_{l=r(t)-v'_u(t)+1}^{r(t)} S_u(l),
\end{align*}
which further implies that $Q_u(l) = 0$ for some $l \in \left[ r(t)-v'_u(t) + 1, r(t) \right]$. This gives us that $B_u(r(t)) \leq v'_u(t).$

Since $v_u(t) + v'_u(t) \leq t/2$, we have $r(t)-v'_u(t) \geq t/2 > w(t)$. Thus, the event $\mathcal{E}_{\ref{ev:no-suboptimal}} \cap \mathcal{E}_{\ref{ev:explore-ub1}}$ implies event \eqref{ev:subopt-schedule-ub1}.

Now, consider the event~\eqref{ev:arrival-capacity-ub1} and note that $A_u(l) -  R_{uk^*_u}(l)$ are i.i.d.\ random variables with mean $-\epsilon_u$ and bounded between $-1$ and $1$. Using Chernoff bound for sum of bounded i.i.d.\ random variables, we have
\begin{align*}
\PP\left[ \sum_{l=r(t)-v'_u(t)+1}^{r(t)} A_u(l) -  R_{uk^*_u}(l) > -\frac{\epsilon_u}{2} v'_u(t) \right]  \leq \exp \left( - \frac{\epsilon_u^2}{8} v'_u(t) \right) \leq \frac{1}{t^3}
\end{align*}
since $v'_u(t) \geq \frac{6 K}{\epsilon_u} w(t) \geq \frac{24}{\epsilon_u^2} \log t$.

Further, by Lemma~\ref{lem:ub-queue1}, $\forall t$ s.t.  $\frac{ w(t)}{\log t} \geq \frac{2}{\epsilon_u}$,
 $$\PP\left[ \left\lbrace Q_u(r(t)-v'_u(t)) > 2 K w(t) \right\rbrace \cap \mathcal{E}_{\ref{ev:no-suboptimal}} \cap \mathcal{E}_{\ref{ev:explore-ub1}} \right] \leq \frac{1}{t^3}.$$

Therefore,  $\forall t$ s.t.  $\frac{ w(t)}{\log t} \geq \frac{2}{\epsilon_u}$ and $v_u(t) + v'_u(t) \leq t/2$, we get
\begin{align*}
& \PP\left[ \left\lbrace B_u\left( t-v_u(t) \right) > v'_u(t) \right\rbrace \cap \mathcal{E}_{\ref{ev:no-suboptimal}} \cap \mathcal{E}_{\ref{ev:explore-ub1}} \right]	\\
& \leq \PP\left[ \sum_{l=r(t)-v'_u(t)+1}^{r(t)} A_u(l) -  R_{uk^*_u}(l) > -\frac{\epsilon_u}{2} v'_u(t) \right] + \PP\left[ \left\lbrace Q_u(r(t)-v'_u(t)) > 2 K w(t) \right\rbrace \cap \mathcal{E}_{\ref{ev:no-suboptimal}} \cap \mathcal{E}_{\ref{ev:explore-ub1}} \right]	\\
& \leq \frac{2}{t^3}.
\end{align*}
\end{proof}

Using the preceding upper bound on the regenerative cycle-length, we derive tighter bounds on the queue-length and the regenerative cycle-length in Lemmas~\ref{lem:ub-queue2} and \ref{lem:busy-period-ub2} respectively. The following lemma is a useful intermediate result. 

\begin{lemma}
\label{lem:conc-arrival-capacity}
For any $u \in [U]$ and $t_2$ s.t. $1 \leq t_2 \leq t,$ $$\PP\left[ \max_{1 \leq s \leq t_2} \left\lbrace \sum_{l=t_2-s+1}^{t_2} A_u(l) - R_{uk^*_u}(l) \right\rbrace \geq \frac{2 \log t}{\epsilon_u} \right] \leq \frac{1}{t^3}.$$
\end{lemma}

\begin{proof}[Proof.]
Let $X_s = \sum_{l=t_2-s+1}^{t_2} A_u(l) - R_{uk^*_u}(l)$. Since $X_s$ is the sum of $s$ i.i.d.\ random variables with mean $\epsilon_u$ and is bounded within $[-1,1]$, Hoeffding's inequality gives
\begin{align*}
\PP\left[ X_s \geq \frac{2 \log t}{\epsilon_u} \right] & = \PP\left[ X_s -\EE X_s \geq \epsilon_u s + \frac{2 \log t}{\epsilon_u} \right]	\\
& \leq \exp\left( -\frac{2\left( \epsilon_u s + \frac{2 \log t}{\epsilon_u} \right)^2}{4 s} \right)	\\
& \leq \exp\left( -4\log t \right),	\\
\end{align*}
where the last inequality follows from the fact that $(a+b)^2 > 4ab$ for any $a, b \geq 0.$ Using union bound over all $1 \leq s \leq t_2$ gives the required result. 
\end{proof}

\begin{lemma}
\label{lem:ub-queue2}
Let $v'_u(t) = \frac{6 K}{\epsilon_u} w(t)$ and $v_u$ be an arbitrary function. Then,
\begin{align*}
\PP\left[ \left\lbrace Q_u(t-v_u(t)) >  \left( \frac{2}{\epsilon_u} + 5 \right) \log t + 30K \frac{v'_u(t) \log^2 t}{t} \right\rbrace \cap \mathcal{E}_{\ref{ev:no-suboptimal}} \cap \mathcal{E}_{\ref{ev:explore-ub1}} \right] \leq  \frac{3}{t^3} + \frac{1}{t^4}
\end{align*}
 $\forall t$ s.t.  $\frac{ w(t)}{\log t} \geq \frac{2}{\epsilon_u}$ and $v_u(t) + v'_u(t) \leq t/2$.
\end{lemma}

\begin{proof}[Proof.]
Let $r(t) = t-v_u(t).$ Now, consider the events 
\begin{align}
B_u(r(t)) & \leq v'_u(t),	\label{ev:busy-period-ub}
\end{align}
\begin{align}
\sum_{l=r(t)-s+1}^{r(t)} A_u(l) - R_{uk^*_u}(l) & \leq \frac{2 \log t}{\epsilon_u}  \; \forall 1 \leq s \leq v'_u(t),	\label{ev:arrival-capacity-ub2}
\end{align}
\begin{align}
\sum_{l=r(t)-v'_u(t)+1}^{r(t)}\mathsf{E}(l) + \sum_{k\neq k^*_u} \mathsf{I}_{uk}(l) \leq 5\log t + 5K \left( \log^3 \left( r(t) \right) - \log^3\left( r(t)-v'_u(t) \right) \right).	\label{ev:subopt-schedule-ub2}
\end{align}

Given the above events, we have
\begin{align*}
Q_u(r(t)) & = \sum_{l = r(t)-B_u(r(t))+1}^{r(t)} A_u(l) - S(l)	\\
& \leq \sum_{l = r(t)-B_u(r(t))+1}^{r(t)}  A_u(l) - R_{uk^*_u}(l) +  \mathsf{E}(l) + \sum_{k\neq k^*} \mathsf{I}_{uk}(l)	\\
& \leq \left( \frac{2}{\epsilon_u} + 5 \right) \log t + 5K \left( \log^3 \left( r(t) \right) - \log^3\left( r(t)-v'_u(t) \right)\right) 	\\
& \leq \left( \frac{2}{\epsilon_u} + 5 \right) \log t + 15K \frac{v'_u(t) \log^2 t}{\left( r(t)-v'_u(t) \right)}	\\
& \leq \left( \frac{2}{\epsilon_u} + 5 \right) \log t + 30K \frac{v'_u(t) \log^2 t}{t},
\end{align*}
where the last inequality is true if $v_u(t) + v'_u(t) \leq t/2.$ 

Using Lemmas~\ref{lem:explore-ub}\ref{lem:explore-ub1}, \ref{lem:busy-period-ub1}, \ref{lem:conc-arrival-capacity}, $\forall t$ s.t.  $\frac{ w(t)}{\log t} \geq \frac{2}{\epsilon_u}$ and $v_u(t) + v'_u(t) \leq t/2$, we get
\begin{align*}
& \PP\left[ \left\lbrace Q_u(t-v_u(t)) >  \left( \frac{2}{\epsilon_u} + 5 \right) \log t + 30K \frac{v'_u(t) \log^2 t}{t} \right\rbrace \cap \mathcal{E}_{\ref{ev:no-suboptimal}} \cap \mathcal{E}_{\ref{ev:explore-ub1}} \right]	\\
& \leq \PP\left[ \left\lbrace B_u\left( t-v_u(t) \right) > v'_u(t) \right\rbrace \cap \mathcal{E}_{\ref{ev:no-suboptimal}} \cap \mathcal{E}_{\ref{ev:explore-ub1}} \right] + \PP\left[ \max_{1 \leq s \leq t_2} \left\lbrace \sum_{l=t_2-s+1}^{t_2} A_u(l) - R_{uk^*_u}(l) \right\rbrace \geq \frac{2 \log t}{\epsilon_u} \right]	\\
& \quad + \PP\left[ \sum_{l=r(t)-v'_u(t)+1}^{r(t)}\mathsf{E}(l) > 5\log t + 5K \left( \log^3 \left( r(t) \right) - \log^3\left( r(t)-v'_u(t) \right) \right) \right]	\\
& \leq  \frac{3}{t^3} + \frac{1}{t^4}.
\end{align*}
\end{proof}

\begin{lemma}
\label{lem:busy-period-ub2}
Let $v'_u(t) = \frac{6 K}{\epsilon_u} w(t)$ and $v_u(t) = \frac{24  \log t}{\epsilon_u^2} + \frac{60 K}{\epsilon_u} \frac{v'_u(t) \log^2 t}{t}.$ Then,
$$\PP\left[ \left\lbrace B_u(t) > v_u(t) \right\rbrace \cap \mathcal{E}_{\ref{ev:no-suboptimal}} \cap \mathcal{E}_{\ref{ev:explore-ub1}} \right]  \leq  \frac{4}{t^3} + \frac{2}{t^4}$$ 
 $\forall t$ s.t.  $\frac{ w(t)}{\log t} \geq \frac{2}{\epsilon_u}$ and $v_u(t) + v'_u(t) \leq t/2$.
\end{lemma}

\begin{proof}[Proof.]
Let $r(t) = t-v_u(t).$ As in Lemma~\ref{lem:busy-period-ub1}, consider the events
\begin{align}
Q_u(r(t)) & \leq \left( \frac{2}{\epsilon_u} + 5 \right) \log t + 30K \frac{v'_u(t) \log^2 t}{t},	\label{ev:queue-small-ub2}
\end{align}
\begin{align}
\sum_{l=r(t)+1}^{t} A_u(l) -  R_{uk^*_u}(l) & \leq -\frac{\epsilon_u}{2} v_u(t),	\label{ev:arrivals-capacity-ub3}
\end{align}
\begin{align}
\sum_{l=r(t)+1}^{t}\mathsf{E}(l) + \sum_{k\neq k^*_u} \mathsf{I}_{uk}(l) & \leq  5\log t + 5K \left( \log^3 t - \log^3\left( r(t) \right) \right).		\label{ev:subopt-schedule-ub3} 
\end{align}
The definition of $v_u(t)$ and events~\eqref{ev:queue-small-ub2}-\eqref{ev:subopt-schedule-ub3} imply that
\begin{align*}
Q_u(r(t)) + \sum_{l=r(t)+1}^{t} A_u(l) & \leq \sum_{l=r(t)+1}^{t} R_{uk^*_u}(l) -  \sum_{l=r(t)+1}^{t}\mathsf{E}(l) + \sum_{k\neq k^*_u} \mathsf{I}_{uk}(l)	\\ 
& \leq \sum_{l=r(t)+1}^{t} S_u(l),
\end{align*}
which further implies that $Q(l) = 0$ for some $l \in \left[ r(t)+ 1, t \right]$ and therefore $B_u(t) \leq v_u(t).$
We can again show that each of the events \eqref{ev:queue-small-ub2}-\eqref{ev:subopt-schedule-ub3} occurs with high probability. Particularly, we can bound the probability of event~\eqref{ev:arrivals-capacity-ub3} in the same way as event~\eqref{ev:arrival-capacity-ub2} in Lemma~\ref{lem:busy-period-ub1} to show that it occurs with probability at least 1-$1/t^3.$ Combining this with Lemmas~\ref{lem:explore-ub}\ref{lem:explore-ub1} and \ref{lem:ub-queue2},  $\forall t$ s.t.  $\frac{ w(t)}{\log t} \geq \frac{2}{\epsilon_u}$ and $v_u(t) + v'_u(t) \leq t/2$, we get 
\begin{align*}
& \PP\left[ \left\lbrace B_u(t) > v_u(t) \right\rbrace \cap \mathcal{E}_{\ref{ev:no-suboptimal}} \cap \mathcal{E}_{\ref{ev:explore-ub1}} \right]	\\
& \leq \PP\left[ \left\lbrace Q_u(r(t)) >  \left( \frac{2}{\epsilon_u} + 5 \right) \log t + 30K \frac{v'_u(t) \log^2 t}{t} \right\rbrace \cap \mathcal{E}_{\ref{ev:no-suboptimal}} \cap \mathcal{E}_{\ref{ev:explore-ub1}} \right]	\\
& \quad + \PP\left[ \sum_{l=r(t)+1}^{t} A_u(l) -  R_{uk^*_u}(l) > -\frac{\epsilon_u}{2} v_u(t) \right]	\\
& \quad + \PP\left[ \sum_{l=r(t)+1}^{t}\mathsf{E}(l) > 5\log t + 5K \left( \log^3 t - \log^3\left( r(t) \right) \right) \right]	\\
& \leq  \frac{4}{t^3} + \frac{2}{t^4}.
\end{align*}
\end{proof}

\begin{proof}[Proof of Theorem~\ref{thm:ucb-best-match-multiq}.]
The proof is based on two main ideas: one is that the regenerative cycle length is not very large, and the other is that the algorithm has correctly identified the optimal matching in late stages.  We combine Lemmas~\ref{lem:exploit-ub} and \ref{lem:busy-period-ub2} to bound the regret at any time $t$ s.t. $t \geq 5.8 \times 10^3$,  $\frac{ w(t)}{\log t} \geq \frac{2}{\epsilon_u}$ and $v_u(t) + v'_u(t) \leq t/2$:
\begin{align}
\Psi_u(t) & = \EE\left[ Q_u(t) - Q^*_u(t) \right] 	\nonumber 	\\
& \leq \EE\left[ (Q_u(t) - Q^*_u(t)) \ind{B_u(t) \leq v_u(t)} + (Q_u(t) - Q^*_u(t)) \ind{B_u(t) > v_u(t)}  \Bigg\vert \mathcal{E}_{\ref{ev:no-suboptimal}} \right] \PP\left[  \mathcal{E}_{\ref{ev:no-suboptimal}} \right]	\nonumber 	\\
& \quad + \EE\left[ Q_u(t) - Q^*_u(t) \Bigg\vert  \mathcal{E}_{\ref{ev:no-suboptimal}}^c \right] \PP\left[  \mathcal{E}_{\ref{ev:no-suboptimal}}^c \right].	\nonumber
\end{align}
Now, the first term on the r.h.s. of the above inequality can be bounded as follows:
\begin{align*}
\EE\left[ (Q_u(t) - Q^*_u(t)) \ind{B_u(t) \leq v_u(t)} \Bigg\vert \mathcal{E}_{\ref{ev:no-suboptimal}} \right] \PP\left[  \mathcal{E}_{\ref{ev:no-suboptimal}} \right]
& \leq  \EE\left[ \sum_{l=t-v_u(t)+1}^{t} \mathsf{E}(l) \right]	\\
& \leq  K\left( \log^3(t) - \log^3(t-v_u(t)) \right)	\\
& \leq 3 K \log^2 t \log \left( 1 + \frac{v_u(t)}{t-v_u(t)} \right)	\\
& \leq 6K\frac{v_u(t) \log^2 t }{t}.
\end{align*}
To bound the rest of the terms in the inequality, we have, for all  $t \geq 5.8 \times 10^3$, ((for Q-UCB,))
\begin{align*}
& \leq \EE\left[ (Q_u(t) - Q^*_u(t)) \ind{B_u(t) > v_u(t)}  \Bigg\vert \mathcal{E}_{\ref{ev:no-suboptimal}} \right] \PP\left[  \mathcal{E}_{\ref{ev:no-suboptimal}} \right] + \EE\left[ Q_u(t) - Q^*_u(t) \Bigg\vert  \mathcal{E}_{\ref{ev:no-suboptimal}}^c \right] \PP\left[  \mathcal{E}_{\ref{ev:no-suboptimal}}^c \right] 	\\
& \leq t \left( \PP\left[ \left\lbrace B_u(t) > v_u(t) \right\rbrace \cap \mathcal{E}_{\ref{ev:no-suboptimal}} \cap \mathcal{E}_{\ref{ev:explore-ub1}} \right] + \PP\left[ \mathcal{E}_{\ref{ev:no-suboptimal}}^c \right] + \PP\left[ \mathcal{E}_{\ref{ev:explore-ub1}}^c \right] \right)	\\
& \leq  t \left( \frac{4}{t^3} + \frac{2}{t^4} + \frac{UK}{6 t^3} + \frac{1}{t^{2K}} \right)	\numberthis	\label{eq:high-prob}	\\
& \leq \frac{24.004 + UK}{6 t^2}.
\end{align*}
Here, inequality~\eqref{eq:high-prob} is obtained using Lemmas~\ref{lem:explore-ub}\ref{lem:explore-ub2}, \ref{lem:exploit-ub} and \ref{lem:busy-period-ub2}. 

The result then follows by combining the above two bounds.
\end{proof}

\begin{proof}[Proof of Corollary~\ref{cor:thm:ucb-best-match-multiq}.]

If $$\log t \geq \max \left\lbrace \frac{4}{\Delta^2 (1 - 1/\beta)^3}, \sqrt{2} \left( \log \frac{2}{\epsilon_u} \right)^{1.5},  \beta \frac{24 K}{\epsilon_u}, \beta \left( \log \log t + \log(15 K^2) \right), \frac{\beta}{\beta - 1} \log\left( \frac{13.2}{K^2 \epsilon_u^2} \right) \right\rbrace,$$ then we can obtain the following.
\begin{enumerate}[label=(\roman*)]
\item Using the fact that  $\log t \geq \sqrt{2} \left( \log \log t \right)^{3/2} \; \forall t \geq 3$ and that  $\log t \geq \sqrt{2} \left( \log \frac{2}{\epsilon_u} \right)^{3/2}$, we have
$$2\log t \geq \sqrt{2} \left( \log \log t \right)^{3/2} + \sqrt{2} \left( \log \frac{2}{\epsilon_u} \right)^{3/2} \geq \left( \log \log t + \log \frac{2}{\epsilon_u} \right)^{3/2},$$ where the last inequality follows from the fact that $\sqrt{2} (a^{3/2} + b^{3/2}) \geq (a + b)^{3/2}$. Combining this with \eqref{eq:lb-w(t)}, we get 
$$\log(w(t)) \geq \left( \frac{2 \log t}{\Delta} \right)^{2/3} \geq \log \log t + \log \frac{2}{\epsilon_u},$$ which gives us that $\frac{w(t)}{\log t} \geq \frac{2}{\epsilon_u}.$
\item $\log t \geq \max \beta \left\lbrace \frac{24 K}{\epsilon_u}, \log \log t + \log(15 K^2) \right\rbrace$ gives us that $\frac{t}{ w(t)} = t^{\frac{1}{\beta}} \geq 15 K^2 \log t$ and $\frac{t}{ w(t)} \geq \frac{24 K}{\epsilon_u}$.
\item Moreover, if $\log t \geq \frac{\beta}{\beta - 1} \log \left( \frac{13.2}{K^2 \epsilon_u^2} \right)$, then $$w(t) = t^{\left( 1 - \frac{1}{\beta} \right)} \geq \frac{13.2}{K^2 \epsilon_u^2},$$ which when combined with  $t \geq 15 K^2 w(t) \log t$ yields
 $$\frac{t}{\log t} \geq 15 K^2 w(t) \geq \frac{198}{\epsilon_u^2}.$$
\end{enumerate}
Therefore, the lower bound on $t$ in Corollory~\ref{cor:thm:ucb-best-match-multiq} yields the following inequalities $\frac{ w(t)}{\log t} \geq \frac{2}{\epsilon_u}$, $\frac{t}{ w(t)} \geq \max\left\lbrace  \frac{24 K}{\epsilon_u}, 15 K^2 \log t \right\rbrace$, and $\frac{t}{\log t} \geq \frac{198}{\epsilon_u^2}$. We now show that these conditions are sufficient for the result in Theorem~\ref{thm:ucb-best-match-multiq}.
\begin{enumerate}[label=(\roman*)]
\item $\frac{t}{ w(t) } \geq \frac{24 K}{\epsilon_u}$ implies that $v'_u(t) \leq \frac{t}{4}$ ,
\item $\frac{t}{ w(t)} \geq 15 K^2 \log t$ implies  that $\frac{24}{\epsilon_u^2} \log t \geq \frac{60 K}{\epsilon_u} \frac{v'_u(t) \log^2 t}{t}$, and therefore $v_u(t) \leq \frac{48}{\epsilon_u^2} \log t$
\item $\frac{t}{\log t} \geq \frac{198}{\epsilon_u^2}$ implies that  $v_u(t) \leq \frac{t}{4}$.
\end{enumerate}
These inequalities when applied to Theorem~\ref{thm:ucb-best-match-multiq} gives the required result, i.e.,
\begin{align*}
\Psi_u(t) & \leq 6K\frac{v_u(t) \log^2 t}{t} + \frac{24.004 + UK}{6 t^2}	\\
 & \leq \frac{288 K \log^3 t }{\epsilon_u^2 t} + \frac{24.004 + UK}{6 t^2}	\\
 & \leq \frac{289 K \log^3 t }{\epsilon_u^2 t}. 
\end{align*} 
\end{proof}

\begin{proof}[Proof of Corollary~\ref{cor:cor:thm:ucb-best-match-multiq}.]
Let $$B = \min \left\lbrace 1 - \frac{1}{\beta} - \frac{1}{\gamma}, \frac{1}{\beta} - \frac{1}{\delta}, 1 - \frac{2}{\gamma} \right\rbrace,$$ and let $C_{\ref{const:ucb-time-lb}}$ be such that 
$$C_{\ref{const:ucb-time-lb}} \geq \max \left\lbrace \tau_2, \exp \left( \frac{4}{\Delta^2 (1 - 1/\beta)^3} \right), 15^{\max\{\gamma, \delta}\} \right\rbrace,$$ 
and $$t^{B} \geq \log t \quad \forall t \geq C_{\ref{const:ucb-time-lb}}.$$

Then $\forall t \geq C_{\ref{const:ucb-time-lb}} \max \left\lbrace \left( \frac{1}{\epsilon_u} \right)^{\gamma},  \left( \frac{K}{\epsilon_u} \right)^{\beta} ,  K^{2\delta} \right\rbrace$, we have the following:
\begin{enumerate}[label=(\roman*)]
\item  Since $B \leq 1 - \frac{1}{\beta} - \frac{1}{\gamma}$ and $C_{\ref{const:ucb-time-lb}} > 2^{\gamma}$,
$$\frac{ w(t)}{\log t} = \frac{t^{1 - \frac{1}{\beta}}}{\log t} \geq \frac{t^{\frac{1}{\gamma}+B}}{\log t} \geq t^{\frac{1}{\gamma}} \geq \frac{2}{\epsilon_u}.$$
\item Since $B \leq \frac{1}{\beta} - \frac{1}{\delta}$ and $C_{\ref{const:ucb-time-lb}} > 15^{\delta}$,
$$\frac{t}{ w(t)} = t^{\frac{1}{\beta}} \geq t^{\frac{1}{\delta} + B} \geq 15 K^2 \log t.$$
Similarly, we can also show that $\frac{t}{ w(t) } \geq \frac{24 K}{\epsilon_u}$.
\item  Since $B \leq 1 - \frac{2}{\gamma}$ and $C_{\ref{const:ucb-time-lb}} \geq 15^{\gamma}$,
$$\frac{t}{\log t} \geq t^{1-B} \geq t^{\frac{2}{\gamma}} \geq \frac{C_{\ref{const:ucb-time-lb}}^{\frac{2}{\gamma}}}{\epsilon_u^2} > \frac{198}{\epsilon_u^2}.$$
\end{enumerate}
As shown in the proof of Corollary~\ref{cor:thm:ucb-best-match-multiq}, the above conditions are sufficient for the upper bound.
\end{proof}

%% file: proof-lb.tex
\section{Proof of Lower Bound for $\alpha$-Consistent Policies}
\label{sec:proofs-lb}
In order to prove Theorems~\ref{thm:lb-late-multiq} and \ref{thm:lb-early-multiq}, we use techniques from existing work in the MAB literature along with some new lower bounding ideas specific to queueing systems. Specifically, we use lower bounds for $\alpha$-consistent policies on the expected number of times a sub-optimal server is scheduled. This lower bound, shown  (in Lemma~\ref{lem:banditlb}) specifically for the problem of scheduling a unique optimal matching, is similar in style to the traditional bandit lower bound by \cite{lai1985asymptotically} but holds in the non-asymptotic setting. Also, as opposed to the traditional change of measure technique used in \cite{lai1985asymptotically}, the proof technique is similar to those used in more recent papers like \cite{bubeck2013bounded, perchet2015batched, combes2015bandits} and uses results from hypothesis testing (Lemma~\ref{lem:tsybakov}). 

\begin{lemma}[\cite{tsybakov2008introduction}]
\label{lem:tsybakov}
Consider two probability measures $P$ and $Q$, both absolutely continuous with respect to a given measure. Then for any event $\mathcal{A}$ we have:
\begin{align*}
P(\mathcal{A}) + Q(\mathcal{A}^c) \geq \frac{1}{2} \exp \{ -\min (\mathrm{KL}(P||Q),\mathrm{KL}(Q||P)) \}.
\end{align*}
\end{lemma}
%

\begin{lemma}
\label{lem:banditlb}
For any problem instance $(\pmb{\lambda}, \pmb{\mu })$ and any $\alpha$-consistent policy, there exist constants $\tau$ and $C$ s.t. for any $u \in [U]$, $k \neq k_u^*$ and $t > \tau$,
\begin{align*}
\mathbb{E} \left[ T_{uk}(t+1)\right] + \sum_{u' \neq u} \mathds{1}\left\lbrace k^*_{u'} = k \right\rbrace \mathbb{E}\left [T_{u'k^*_u} (t+1)\right] \geq \frac{1}{\mathrm{KL}\left(\mu_{min}, \frac{\mu^*+1}{2} \right)} \left(  (1 - \alpha) \log t - \log(4KC)  \right) .
\end{align*}
\end{lemma}

\begin{proof}[Proof.]
Without loss of generality, let the optimal servers for the $U$ queues be denoted by the first $U$ indices. In other words, a server $k > U$ is not an optimal server for any queue, i.e., for any $u' \in [U]$, $K \geq k > U$, $\mathds{1}\left\lbrace k^*_{u'} = k \right\rbrace = 0$. Also, let $\beta = \frac{\mu^*+1}{2}$.  

We will first consider the case $k \leq U$. For a fixed user $u$ and server $k \leq U$, let $u'$ be the queue that has $k$ as the best server, i.e., $k_{u'}^* = k$. Now consider the two problem instances $(\pmb{\lambda}, \pmb{\mu })$ and $(\pmb{\lambda}, \pmb{\hat{\mu}})$, where $\pmb{\hat{\mu}}$ is the same as $\pmb{\mu }$ except for the two entries corresponding to indices $(u,k), (u', k^*_u)$ replaced by $\beta$. Therefore, for the problem instance $(\pmb{\lambda}, \pmb{\hat{\mu}})$, the best servers are swapped for queues $u$ and $u'$ 
 and remain the same for all the other queues. Let $\PP ^t _{\pmb{\mu }}$ and $\PP ^t _{\pmb{\hat{\mu}}}$ be the distributions corresponding to the arrivals, chosen servers and service obtained in the first $t$ plays for the two instances under a fixed $\alpha$-consistent policy. Recall that $T_{uk}(t+1) = \sum_ {s = 1} ^t \mathds{1}\lbrace \kappa_u(s) = k \rbrace \; \forall u \in [U], k \in [K]$. Define the event $ \mathcal{A} =  \{ T_{uk}(t+1) > t/2 \} $. By the definition of $\alpha$-consistency there exists a fixed integer $\tau$ and a fixed constant $C$ such that for all $t > \tau$ we have, 
\begin{align*}
\mathbb{E}^t _{\pmb{\mu }} \left[ \sum_ {s = 1} ^t \mathds{1}\lbrace \kappa_u(s) = k\rbrace\right] &\leq C t^{\alpha} \\
\mathbb{E}^t _{\pmb{\hat{\mu}}} \left[ \sum_ {s = 1} ^t \mathds{1}\lbrace\kappa_u(s) = k'\rbrace\right] &\leq C t^{\alpha} \mbox{  , } \forall k' \neq k.
\end{align*}
A simple application of Markov's inequality yields
\begin{align*}
 \PP^t_{\pmb{\mu }}(\mathcal{A}) &\leq \frac{2C}{t^{1-\alpha}} \\
 \PP^t_{\pmb{\hat{\mu}}} (\mathcal{A}^c) &\leq \frac{2C(K-1)}{t^{1-\alpha}}. 
\end{align*}

We can now use Lemma~\ref{lem:tsybakov} to conclude that
\begin{equation}
\mathrm{KL}(\PP ^t _{\pmb{\mu }}||\PP ^t _{\pmb{\hat{\mu}}} ) \geq (1 - \alpha) \log t - \log(4KC). \label{eq:lbmain}
\end{equation} 
It is, therefore, sufficient to show that
\begin{align*}
\mathrm{KL}\left( \PP ^t _{\pmb{\mu }} || \PP ^t _{\pmb{\hat{\mu}}} \right) = \mathrm{KL}\left(\mu_{uk}, \beta \right)\mathbb{E}^t_{\pmb{\mu }} [T_{uk} (t+1)] + \mathrm{KL}\left(\mu_{u'k^*_u}, \beta \right)\mathbb{E}^t_{\pmb{\mu }} [T_{u'k^*_u} (t+1)].
\end{align*}
For the sake of brevity we write the scheduling sequence in the first $t$ time-slots $\{\pmb{\kappa}(1),\pmb{\kappa}(2),...,\pmb{\kappa}(t)\}$ as $\pmb{\kappa}^{(t)}$, and similarly we define $\mathbf{A}^{(t)}$ as the number of arrivals to the queue and $\mathbf{S}^{(t)}$ as the service offered by the scheduled servers in the first $t$ time-slots. Let $\mathbf{Z}^{(t)} = (\pmb{\kappa}^{(t)}, \mathbf{A}^{(t)}, \mathbf{S}^{(t)})$. 
The $\mathrm{KL}$-divergence term can now be written as
 \begin{align*}
\mathrm{KL}(\PP ^t _{\pmb{\mu }}||\PP ^t _{\pmb{\hat{\mu}}} ) = \mathrm{KL}(\PP^t _{\pmb{\mu}}(\mathbf{Z}^{(t)} ) || \PP^t  _{\pmb{\hat{\mu}}}(\mathbf{Z}^{(t)} )).
 \end{align*}
 We can apply the chain rule of divergence to conclude that
 \begin{align*}
\mathrm{KL}(\PP^t _{\pmb{\mu }}(\mathbf{Z}^{(t)} ) || \PP^t  _{\pmb{\hat{\mu}}}(\mathbf{Z}^{(t)} )) & = \mathrm{KL}(\PP^t _{\pmb{\mu }}(\mathbf{Z}^{(t-1)} ) || \PP^t  _{\pmb{\hat{\mu}}}(\mathbf{Z}^{(t-1)} ) ) \\
  & \quad +  \mathrm{KL}(\PP^t _{\pmb{\mu }}(\pmb{\kappa}(t) \given \mathbf{Z}^{(t-1)} ) || \PP^t  _{\pmb{\hat{\mu}}}(\pmb{\kappa}(t) \given \mathbf{Z}^{(t-1)} ) ) \\
  & \quad + \mathbb{E}^t _{\pmb{\mu }} \left[ \mathds{1}\lbrace \kappa_u(t) = k \rbrace \mathrm{KL}\left(\mu_{uk}, \beta \right) + \mathds{1}\lbrace \kappa_{u'}(t) = k^*_u \rbrace \mathrm{KL}\left(\mu_{u'k^*_u}, \beta \right)\right].
 \end{align*}
 We can apply this iteratively to obtain
 \begin{align*}
  \mathrm{KL}(\PP ^{t} _{\pmb{\mu }} || \PP ^{t} _{\pmb{\hat{\mu}}})   & = \sum _{s = 1}^{t} \mathbb{E}^t_ {\pmb{\mu }} \left[  \mathds{1}\lbrace \kappa_u(s) = k \rbrace \mathrm{KL}\left(\mu_{uk}, \beta \right) \right] 	\\
    &  \quad +  \sum _{s = 1}^{t} \mathbb{E}^t _{\pmb{\mu }} \left[  \mathds{1}\lbrace \kappa_{u'}(s) = k^*_u \rbrace \mathrm{KL}\left(\mu_{u'k^*_u}, \beta \right)\right] \\
  & \quad + \sum _{l = 1} ^{t} \mathrm{KL}(\PP^t _{\pmb{\mu }}(\pmb{\kappa}(l) \given \mathbf{Z}^{(l-1)}) || \PP^t  _{\pmb{\hat{\mu}}}(\pmb{\kappa}(l) \given \mathbf{Z}^{(l-1)}) ) \numberthis \label{eq:tsybakov}
 \end{align*}
Note that the second summation in \eqref{eq:tsybakov} is zero, as over a sample path the policy pulls the same servers irrespective of the parameters. Therefore, we obtain
\begin{equation*}
\mathrm{KL}(\PP ^{t} _{\pmb{\mu }} || \PP ^{t} _{\pmb{\hat{\mu}}})  = \mathrm{KL}\left(\mu_{uk}, \beta \right)\mathbb{E}^t_{\pmb{\mu }} [T_{uk} (t+1)] + \mathrm{KL}\left(\mu_{u'k^*_u}, \beta \right)\mathbb{E}^t_{\pmb{\mu }} [T_{u'k^*_u} (t+1)] ,
\end{equation*}
which can be substituted in \eqref{eq:lbmain} to obtain the required result for $K \leq U$.

Now, consider the case $k > U$, where  $\sum_{u \in U}  \mathds{1}\left\lbrace k^*_{u} = k \right\rbrace = 0$. We again compare the two problem instances  $(\pmb{\lambda}, \pmb{\mu })$ and $(\pmb{\lambda}, \pmb{\hat{\mu}})$, where $\pmb{\hat{\mu}}$ is the same as $\pmb{\mu }$ except for the entry corresponding to index $(u,k)$ replaced by $\beta$. Therefore, for the problem instance $(\pmb{\lambda}, \pmb{\hat{\mu}})$, the best server for user $u$ is server $k$ while the best servers for all other queues remain the same. We can again use the same technique as before to obtain
\begin{equation*}
\mathrm{KL}(\PP ^{t} _{\pmb{\mu }} || \PP ^{t} _{\pmb{\hat{\mu}}})  = \mathrm{KL}\left(\mu_{uk}, \beta \right)\mathbb{E}^t_{\pmb{\mu }} [T_{uk} (t+1)],
\end{equation*}
which, along with \eqref{eq:lbmain}, gives the required result for $K > U$. 
\end{proof}
As a corollary of the above result, we now derive lower bound on the total expected number of sub-optimal schedules summed across all queues. In addition, we also show, for each individual queue, a lower bound for those servers which are sub-optimal for all the queues. As in the proof of Lemma~\ref{lem:banditlb}, we assume without loss of generality that the first $U$ indices denote the optimal servers for the $U$ queues. 
\begin{corollary}
\label{cor:lem:banditlb}
For any problem instance $(\pmb{\lambda}, \pmb{\mu })$ and any $\alpha$-consistent policy, there exist constants $\tau$ and $C$ s.t. for any $t > \tau$,
\begin{enumerate}[label=(\alph*)]
\item \label{item:banditlb-avg} 
\begin{align*}
2\Delta \sum_{u \in [U]} \sum_{k \neq k^*_u} \mathbb{E} \left[ T_{uk}(t+1)\right]	 \geq U (K-1) D(\pmb{\mu}) \left(  (1 - \alpha) \log t - \log(4KC)  \right),
\end{align*}
\item \label{item:banditlb-single1} for any $u \in [U]$, 
\begin{align*}
2 \Delta \sum_{k \neq k^*_u} \mathbb{E} \left[ T_{uk}(t+1)\right] \geq (U-1) D(\pmb{\mu}) \left(  (1 - \alpha) \log t - \log(4KC)  \right),
\end{align*}
\item \label{item:banditlb-single2} and for any $u \in [U]$, 
\begin{align*}
\Delta \sum_{k > U} \mathbb{E} \left[ T_{uk}(t+1)\right] \geq (K-U) D(\pmb{\mu}) \left(  (1 - \alpha) \log t - \log(4KC)  \right),
\end{align*}
\end{enumerate}
where $D(\pmb{\mu})$ is given by \eqref{eq:constD}.
\end{corollary}
\begin{proof}[Proof.]
To prove part~\ref{item:banditlb-avg}, we observe that a unique optimal server for each queue in the system implies that 
\begin{align*}
\sum_{u \in [U]} \sum_{k \neq k^*_u} \mathbb{E} \left[ T_{uk}(t+1)\right] 	
& \geq \sum_{u \in [U]} \sum_{u' \neq u}  \mathbb{E} \left[ T_{uk^*_{u'}}(t+1)\right]	\\
& = \sum_{u \in [U]} \sum_{k \neq k^*_u}  \sum_{u' \neq u} \mathds{1}\left\lbrace k^*_{u'} = k \right\rbrace \mathbb{E}\left [T_{u'k^*_u} (t+1)\right].
\end{align*}
Now, from Lemma~\ref{lem:banditlb}, there exist constants $C$ and $\tau$ such that for $t > \tau$,
\begin{align*}
2\sum_{u \in [U]} \sum_{k \neq k^*_u} \mathbb{E} \left[ T_{uk}(t+1)\right] 
& \geq \sum_{u \in [U]} \sum_{k \neq k^*_u} \left( \mathbb{E} \left[ T_{uk}(t+1)\right] + \sum_{u' \neq u} \mathds{1}\left\lbrace k^*_{u'} = k \right\rbrace \mathbb{E}\left [T_{u'k^*_u} (t+1)\right] \right)	\\
& \geq \frac{U (K-1)}{\mathrm{KL}\left(\mu_{min}, \frac{\mu^*+1}{2} \right)} \left(  (1 - \alpha) \log t - \log(4KC)  \right).
\end{align*}
Using the definition of $D(\pmb{\mu})$ in the above inequality gives part~\ref{item:banditlb-avg} of the corollary.

 To prove part~\ref{item:banditlb-single1}, we can assume without loss of generality that a perfect matching is scheduled in every time-slot. Using this, and the fact that any server is assigned to at most one queue in every time-slot, for any $u \in [U]$, we have
 \begin{align*}
 T_{uk^*_u}(t+1) + \sum_{k \neq k^*_u} T_{uk}(t+1) = t \geq   T_{uk^*_u}(t+1) + \sum_{u' \neq u} T_{u'k^*_u} (t+1), 
\end{align*}  
which gives us
\begin{align}
\label{eq:bandit-samplepath-lb}
 \sum_{k \neq k^*_u} T_{uk}(t+1) \geq \max \left\lbrace \sum_{u' \neq u} T_{uk^*_{u'}} (t+1),  \sum_{u' \neq u} T_{u'k^*_u} (t+1)\right\rbrace.
\end{align}
From Lemma~\ref{lem:banditlb} we have, for any $u' \neq u$ and  for $t > \tau$,
\begin{align*}
\mathbb{E}\left [T_{uk^*_{u'}} (t+1)\right] + \mathbb{E}\left [T_{u'k^*_u} (t+1)\right] \geq \frac{1}{\mathrm{KL}\left(\mu_{min}, \frac{\mu^*+1}{2} \right)} \left(  (1 - \alpha) \log t - \log(4KC)  \right),
\end{align*}
which gives
\begin{align*}
\sum_{u' \neq u} \mathbb{E}\left [T_{uk^*_{u'}} (t+1)\right] + \mathbb{E}\left [T_{u'k^*_u} (t+1)\right] \geq \frac{U-1}{\mathrm{KL}\left(\mu_{min}, \frac{\mu^*+1}{2} \right)} \left(  (1 - \alpha) \log t - \log(4KC)  \right).
\end{align*}
Combining the above with \eqref{eq:bandit-samplepath-lb}, we have  for $t > \tau$
\begin{align*}
\sum_{k \neq k^*_u} \mathbb{E} \left[ T_{uk}(t+1)\right]
& \geq \max \left\lbrace \sum_{u' \neq u} \mathbb{E}\left [T_{uk^*_{u'}} (t+1)\right], \sum_{u' \neq u} \mathbb{E}\left [T_{u'k^*_u} (t+1)\right] \right\rbrace 	\\
& \geq  \frac{U-1}{2 \mathrm{KL}\left(\mu_{min}, \frac{\mu^*+1}{2} \right)} \left(  (1 - \alpha) \log t - \log(4KC)  \right).
\end{align*}
 To prove part~\ref{item:banditlb-single2}, we use the fact that $\mathds{1}\left\lbrace k^*_{u'} = k \right\rbrace = 0$ for any $u' \in [U]$, $K \geq k > U$. Therefore, for $t > \tau$, we have
\begin{align*}
\sum_{k > U} \mathbb{E} \left[ T_{uk}(t+1)\right] 
& = \sum_{k > U} \left( \mathbb{E} \left[ T_{uk}(t+1)\right] + \sum_{u' \neq u} \mathds{1}\left\lbrace k^*_{u'} = k \right\rbrace \mathbb{E}\left [T_{u'k^*_u} (t+1)\right] \right)	\\
& \geq \frac{K-U}{\mathrm{KL}\left(\mu_{min}, \frac{\mu^*+1}{2} \right)} \left(  (1 - \alpha) \log t - \log(4KC)  \right),
\end{align*}
which gives the required result. 
\end{proof}

\subsection{Late Stage: Proof of Theorem~\ref{thm:lb-late-multiq}}
\label{subsec:proof-lb-late}
The following lemma, which gives a lower bound on the queue-regret in terms of probability of sub-optimal schedule in a single time-slot, is the key result used in the proof of Theorem~\ref{thm:lb-late-multiq}. The proof for this lemma is based on the idea that the growth in regret in a single-time slot can be lower bounded in terms of the probability of sub-optimal schedule in that time-slot.
\begin{lemma}
\label{lem:lb-queue1}
For any problem instance characterized by $(\pmb{\lambda}, \pmb{\mu}),$ and for any scheduling policy, and user $u \in [U]$,
\begin{align*}
\Psi_u(t) \geq \lambda_u\sum_{k \neq k^*_u} \Delta_{uk} \PP \left[ \mathds{1}\lbrace \kappa_u(t) = k \rbrace = 1 \right].
\end{align*}
\end{lemma}
\begin{proof}[Proof.]
%
We lower bound the queue-regret for queue $u$ for the alternate coupled service process described in Section~\ref{subsec:proof-easy-ub}. As seen in the proof of Lemma~\ref{lem:easy-ub}, since $\mu^*_u > \mu_{uk} \; \forall k \neq k^*_u$, for the alternate system, we have $R_{uk^*_u}(t) \geq R_{uk}(t) \; \forall k \neq k^*_u, \, \forall t$. This implies that $Q^*_u(t) \leq Q_u(t) \; \forall t$. Now, for any given $t$, using the fact that $Q^*_u(t-1) \leq Q_u(t-1)$, it is easy to see that 
\begin{align*}
Q_u(t) - Q^*_u(t) \geq \mathds{1}\left\lbrace A_u(t) = 1 \right\rbrace \left( R_{k^*_u}(t) - \sum_{k=1}^K \mathds{1}\lbrace \kappa_u(t) = k \rbrace R_{uk}(t) \right).
\end{align*}
Therefore,
\begin{align*}
\mathbb{E} \left[ Q_u(t) - Q^*_u(t) \right] 
& \geq \mathbb{E} \left[ \mathds{1}\lbrace A_u(t) = 1 \rbrace \left( R_{k^*_u}(t) - \sum_{k=1}^K \mathds{1}\lbrace \kappa_u(t) = k \rbrace R_{uk}(t) \right) \right]	\\
& = \lambda_u \sum_{k \neq k^*_u} \PP \left[ \mathds{1}\lbrace \kappa_u(t) = k \rbrace = 1 \right] \PP\left[ \mu_{uk} < U(t) \leq \mu^*_u \right]	\\
& = \lambda_u \sum_{k \neq k^*_u} \Delta_{uk} \PP \left[ \mathds{1}\lbrace \kappa_u(t) = k \rbrace = 1 \right].
\end{align*}
\end{proof}

We now use Lemma~\ref{lem:lb-queue1} in conjunction with the lower bound for the expected number of sub-optimal schedules for an $\alpha$-consistent policy (Corollary~\ref{cor:lem:banditlb}) to prove Theorem~\ref{thm:lb-late-multiq}.
\begin{proof}[Proof of Theorem~\ref{thm:lb-late-multiq}.]
From Lemma~\ref{lem:lb-queue1} we have,
\begin{align}
\label{eq:lb-avg}
 \Psi_u(t) & \geq \lambda_u \sum_{k \neq k^*_u} \Delta_{uk} \PP \left[ \mathds{1}\lbrace \kappa_u(t) = k \rbrace = 1 \right]	\nonumber	\\
& \geq \lambda_{min} \Delta \sum_{k \neq k^*_u} \PP \left[ \mathds{1}\lbrace \kappa_u(t) = k \rbrace = 1 \right].
\end{align}
Therefore,
\begin{align*}
& \sum_{s=1}^t \sum_{u \in [U]}  \Psi_u(s) \geq \lambda_{min} \Delta \sum_{u \in [U]} \sum_{k \neq k^*_u} \mathbb{E} \left[ T_{uk}(t+1)\right].
\end{align*}
 We now claim that
\begin{align}
\label{eq:one-step-lb}
\sum_{u \in [U]}  \Psi_u(t) \geq \frac{U(K-1)}{8t} \lambda_{min} D(\pmb{\mu}) (1 - \alpha)
\end{align}
for infinitely many $t$. This follows from part~\ref{item:banditlb-avg} of Corollary~\ref{cor:lem:banditlb} and the following fact:
\begin{fact}
For any bounded sequence $\left\lbrace a_n \right\rbrace$, if there exist constants $C > 0$ and $N_0 \in \setN$ such that $\sum_{m=1}^n a_m \geq C\log n \; \forall n \geq N_0$, then $a_n \geq \frac{C}{2n}$ infinitely often.
\end{fact}
\begin{proof}
Proof by contradiction: Let $B$ be an upper bound on the sequence. Suppose $\sum_{m=1}^n a_m \geq C\log n \; \forall n \geq N_0$ and there exists $N_1 \in \setN$ such that  $a_n \leq \frac{C}{2n}$ for all $n > N_1$. Wlog let $N_1$ satisfy the inequality $\exp\left( \frac{2 B N_1}{C} + 1 \right) \geq N_1$. Then for any $n > \exp\left( \frac{2 B N_1}{C} + 1 \right)$, we have
\begin{align*}
\sum_{m=1}^n a_m 
 \leq B N_1 + \sum_{m=N_1+1}^n \frac{C}{2n}	\leq B N_1 +  \frac{C}{2}(\log n + 1)
 < C\log n,
\end{align*}
which is a contradiction to our hypothesis that $\sum_{m=1}^n a_m \geq C\log n \; \forall n \geq N_0$. This proves the claim.
\end{proof}

Similarly, for any $u \in U$, it follows from parts~\ref{item:banditlb-single1} and \ref{item:banditlb-single2} of Corollary~\ref{cor:lem:banditlb} that
\begin{align}
 \Psi_u(t)  \geq \frac{\max \left\lbrace U-1, 2(K-U) \right\rbrace}{8t} \lambda_{min} D(\pmb{\mu}) (1 - \alpha)
\end{align}
for infinitely many $t$. 
\end{proof}

\subsection{Early Stage: Proof of Theorem~\ref{thm:lb-early-multiq}}
\label{subsec:proof-lb-early}
In order to prove Theorem~\ref{thm:lb-early-multiq}, we first derive, in the following lemma, a lower bound on the queue-regret in terms of the expected number of sub-optimal schedules. 
\begin{lemma}
\label{lem:lb-queue2}
For any system with parameters $(\pmb{\lambda}, \pmb{\mu})$, any policy, and any user $u \in [U]$, the regret is lower bounded by
\begin{align*}
\Psi_u(t) \geq \sum_{k \neq k^*_u} \Delta_{uk} \mathbb{E} \left[ T_{uk}(t+1)\right] - \epsilon_u t.
\end{align*}
\end{lemma}

\begin{proof}[Proof.]
 Since $Q_u(0) \sim \pi_{\lambda_u,\mu^*_u}$, we have, 
\begin{align*}
\Psi_u(t) & = \mathbb{E} \left[ Q_u(t) - Q^*_u(t) \right] \\
        & = \mathbb{E} \left[ Q_u(t) - Q_u(0) \right] \\
        & \geq \mathbb{E} \left[ \sum _{l = 1} ^t A_u(l) - S_u(l) \right] \\
        & = \lambda_u t - \sum_{k = 1} ^{K} \mathbb{E} \left[ T_{uk}(t+1)\right]\mu _{uk} \\
        & = \lambda_u t - \left(t -   \sum_{k \neq k^*_u}  \mathbb{E} \left[ T_{uk}(t+1)\right]\right) \mu*_u -  \sum_{k \neq k^*_u}  \mathbb{E} \left[ T_{uk}(t+1)\right] \mu _{uk} \\
        & =  \sum_{k \neq k^*_u}  \Delta_{uk} \mathbb{E} \left[ T_{uk}(t+1)\right] - \epsilon_u t .
\end{align*}
\end{proof}
We now use this lower bound along with the lower bound on the expected number of sub-optimal schedules for $\alpha$-consistent policies (Corollary~\ref{cor:lem:banditlb}).
\begin{proof}[Proof of Theorem~\ref{thm:lb-early-multiq}.]
To prove part~\ref{item:lb-early-avg} of the theorem, we use Lemma~\ref{lem:lb-queue2} and part~\ref{item:banditlb-avg} of corollary~\ref{cor:lem:banditlb} as follows:
For any $\gamma > \frac{1}{1-\alpha}$, there exist constants $C_{\ref{const:switch-lb}}$ and $\tau$ such that for all $t \in [\max \{C_{\ref{const:switch-lb}} K^{\gamma }, \tau \}, (K-1)\frac{D(\pmb{\mu})}{4\bar{\epsilon}}]$, 
\begin{align*}
\frac{1}{U}\sum_{u \in [U]} \Psi_u(t) & \geq \frac{\Delta}{U}\sum_{u \in [U]} \left( \sum_{k \neq k^*_u}  \mathbb{E} \left[ T_{uk}(t+1)\right] - \epsilon_u t \right) 	\\
& \geq (K-1)\frac{D(\pmb{\mu})}{2} \left( (1 - \alpha)\log t - \log(KC_{\ref{const:switch-lb}}) \right) - \bar{\epsilon} t \\
		& \geq (K-1)\frac{D(\pmb{\mu})}{2} \frac{\log t}{ \log \log t} - \bar{\epsilon} t \\
		& \geq (K-1)\frac{D(\pmb{\mu})}{4}\frac{\log t}{ \log \log t},
\end{align*}
where the last two inequalities follow since $t \geq C_{\ref{const:switch-lb}} K^{\gamma }$ and $t \leq (K-1)\frac{D(\pmb{\mu})}{4\bar{\epsilon}}$.

Part~\ref{item:lb-early-single} of the theorem can be similarly shown using parts~\ref{item:banditlb-single1} and \ref{item:banditlb-single2} of Corollary~\ref{cor:lem:banditlb}. 
\end{proof}